\documentclass{article}
\usepackage{amsmath}      % 数学排版
\usepackage{amssymb}      % 数学符号
\usepackage{amsthm}       % 定理环境
\usepackage{mathrsfs}     % 花体字母
\usepackage{bm}           % 粗体数学符号
\usepackage{color}        % 颜色支持
\usepackage{graphicx}     % 图形支持
\usepackage{float}        % 浮动体支持
\usepackage{subfigure}    % 子图支持
\usepackage{geometry}     % 页面设置
\usepackage{hyperref}     % 超链接支持
\usepackage{comment}
\usepackage{appendix}
\usepackage{chngcntr} % 引入控制计数器的包
\usepackage{enumitem}
\newtheorem{thm}{Theorem}[section]
\newtheorem{lemma}[thm]{Lemma}
\newtheorem{definition}[thm]{Definition}

\newtheorem{remark}[thm]{Remark}
\newcommand{\inner}[1]{\langle #1 \rangle}
\newcommand{\abs}[1]{\lvert#1\rvert}
\newcommand{\sgn}{{\rm sign}}

\newcommand{\argmin}{\operatorname*{argmin}}

\newcommand\extrafootertext[1]{%
    \bgroup
    \renewcommand\thefootnote{\fnsymbol{footnote}}%
    \renewcommand\thempfootnote{\fnsymbol{mpfootnote}}%
    \footnotetext[0]{#1}%
    \egroup
}

\numberwithin{equation}{section}
\title{On the Performance of Amplitude-Based Models for Low-Rank Matrix Recovery}
\author{%
  Huanmin~Ge and Zhiqiang~Xu
}
\date{}

\begin{document}

\maketitle

%\maketitle
\begin{abstract}
In this paper, we focus on low-rank phase retrieval, which aims to reconstruct a matrix $\bm{X}_0\in \mathbb{R}^{n\times m}$ with ${\rm rank}(\bm{X}_0)\le r$ from noise-corrupted amplitude measurements $\bm{y}=|\mathcal{A}(\bm{X}_0)|+\bm{\eta}$, where $\mathcal{A}:\mathbb{R}^{n\times m}\rightarrow \mathbb{R}^{p}$ is a linear map and $\bm{\eta}\in \mathbb{R}^p$ is the noise vector. We first examine the rank-constrained nonlinear least-squares model $\hat{\bm{X}}\in \mathop{\argmin}\limits_{\substack{\bm{X}\in \mathbb{R}^{n\times m},\mathrm{rank}(\bm{X})\le r}}\||\mathcal{A}(\bm{X})|-\bm{y}\|_2^2$ to estimate $\bm{X}_0$, and demonstrate that the reconstruction error satisfies
$\min\{\|\hat{\bm{X}}-\bm{X}_0\|_F, \|\hat{\bm{X}}+\bm{X}_0\|_F\}\lesssim \frac{\|\bm{\eta}\|_2}{\sqrt{p}}$
with high probability, provided $\mathcal{A}$ is a Gaussian measurement ensemble and $p\gtrsim (m+n)r$. We also prove that the error bound $\frac{\|\bm{\eta}\|_2}{\sqrt{p}}$ is tight up to a constant. Furthermore, we relax the rank constraint to a nuclear-norm constraint. Hence, we propose the Lasso model for low-rank phase retrieval, i.e., the constrained nuclear-norm model and the unconstrained version. We also establish comparable theoretical guarantees for these models. To achieve this, we introduce a strong restricted isometry property (SRIP) for the linear map $\mathcal{A}$, analogous to the strong RIP in phase retrieval.
This work provides a unified treatment that extends existing results in both phase retrieval and  low-rank matrix recovery from rank-one measurements.
\end{abstract}

\extrafootertext{H.~Ge is with the School of Sports Engineering, Beijing Sport University, Beijing, 100084, China (e-mail: gehuanmin@bsu.edu.cn). }
\extrafootertext{Z.~Xu (Corresponding author) is with   State Key Laboratory of Mathematical Sciences, Academy of Mathematics and Systems Science, Chinese Academy of Sciences, Beijing 100190, China;    School of Mathematical Sciences, University of Chinese Academy of Sciences, Beijing 100049, China. (e-mail: xuzq@lsec.cc.ac.cn).}
\extrafootertext{Huanmin Ge  supported by the NSFC (No.12371094), Beijing Natural Science Foundation (No.1232020). Zhiqiang Xu is supported by the National
Science Fund for Distinguished Young Scholars (12025108) and NSFC (12471361, 12021001, 12288201).}

\begin{center}
\section{Introduction}
\end{center}

\subsection{Problem Setup}

\subsubsection{Phase retrieval}
Phase retrieval  has received much recent attention in the applied mathematics and statistics, with numerous applications spanning fields such as X-ray crystallography \cite{haeusele2023advanced, pfeiffer2006phase}, coherent diffraction imaging \cite{bacca2019super, song2022fast}, and quantum mechanics \cite{corbett2006pauli}.
The  goal of phase retrieval is to recover the target signal $\bm{x}_0\in \mathbb{F}^{n}$ with $\mathbb{F}\in\{\mathbb{R}, \mathbb{C}\}$, from  the magnitude of linear observations
\begin{align*}%\label{PRv}
y_i=\abs{\langle\bm{a}_i,\bm{x}_0\rangle}+\eta_i,\ \ \ \ \ \ i\in [p]:=\{1,2,\ldots,p\}
\end{align*}
where $\bm{a}_i\in \mathbb{F}^n $ are the given measurement vectors and $\bm{\eta}=(\eta_1, \eta_2 , \ldots , \eta_p)^{\top}\in \mathbb{R}^{p}$ is a noise vector.
Here, $^\top$ denotes the operation of taking the transpose of a matrix or vector.
 An  intuitive method to recovery  $\bm{x}_0$ is  through the nonlinear least squares model
\begin{align}\label{pro:squavector}
\argmin_{\bm{x}\in \mathbb{F}^{n}}\||\bm{A}\bm{x}|-\bm{y}\|_2^2,
\end{align}
where $\bm{A}:=[\bm{a}_1, \ldots, \bm{a}_p]^{*}\in \mathbb{F}^{p\times n}$ and $\bm{y}:=(y_1,\ldots,y_p)^{\top}$.
Throughout this paper, $^*$
  denotes the conjugate transpose when $\mathbb{F} = \mathbb{C}$ and the transpose when $\mathbb{F} = \mathbb{R}$.
%Many efficient algorithms for solving the amplitude-based model (\ref{pro:squavector}), as well as
Numerous theoretical results on phase retrieval, have been extensively introduced and studied in recent years.
Particularly, the problem of finding the minimal measurement number for phase retrieval
 has received much recent attention \cite{balan2006signal, conca2015algebraic, wang2019generalized}. Concurrently,
the estimation performance of the nonlinear least squares model \eqref{pro:squavector} for phase retrieval has been investigated and further developed (see \cite{huang2020estimation,xia2024performance}).

\subsubsection{Low-rank phase retrieval}
In this paper, we focus on low-rank phase retrieval, where the goal is to estimate a low-rank matrix $\bm{X}_0\in {\mathbb R}^{n\times m}$ from noisy amplitude-based measurements, which are described by
\begin{align*}
\bm{y}=|\mathcal{A}(\bm{X}_0)|+\bm{\eta},
% \ \ \ \ (\mathrm{i.e.}, y_i=|\langle \bm{A}_i, \bm{X}_0\rangle|+{\eta}_i, i\in [p]).
\end{align*}
Here, $\mathcal{A}:{\mathbb R}^{n\times m}\rightarrow {\mathbb R}^p$ is a  linear operator defined as
\[
{\mathcal A}(\bm{X}_0)=(\langle \bm{A}_1, \bm{X}_0\rangle, \langle \bm{A}_2, \bm{X}_0\rangle,\ldots,\langle \bm{A}_p, \bm{X}_0\rangle)^{\top}.
\]
In this context,  $\langle \bm{A}_i, \bm{X}_0 \rangle = \mathrm{Tr}(\bm{A}_i^{\top}\bm{X}_0)$ denotes the standard inner product, with $\mathrm{Tr}(\cdot)$ being the trace operator.

Low-rank phase retrieval (LRPR), initially introduced in \cite{lowrank21}, received further attention in \cite{lee2021phase,pmlr-v97-nayer19a, 9050801, 9537806, 9912351}.
Low-rank phase retrieval  has emerged as a pivotal technique for facilitating rapid and cost-effective dynamic phaseless imaging across various domains, including
dynamic astronomical imaging \cite{butala2007monte} and Fourier ptychography \cite{chen2020fast, eckert2016algorithmic, jagatap2019sample}.

 %The fundamental objective of LRPR is to accurately and robustly reconstruct an \( n \times m \) matrix \( \bm{X}_0 \) of rank \( r \) from phaseless measurements.

It is noteworthy that  the amplitude-based low-rank phase retrieval formulation subsumes a broad class of phase retrieval problems, including several well-studied special cases:

\begin{enumerate}[label=(\roman*)]
\item When $m=1$, specifically with $\bm{X}_0=\bm{x}_0\in \mathbb{R}^{n}$ and $\bm{A}_i=\bm{a}_i\in \mathbb{R}^n$ for all $i\in [p]$, the problem reduces to the classical vector phase retrieval problem.
\item When $m = n$, with symmetric $\bm{A}_i\in \mathbb{R}^{n\times n}$ and $\bm{X}_0 = \bm{x}_0 \bm{x}_0^{\top} \in \mathbb{R}^{n \times n}$, the problem simplifies to the generalized phase retrieval problem \cite{ huang2021almost,wang2019generalized}.
\item In cases where $\bm{A}_i = \bm{a}_i \bm{a}_i^{\top}\in \mathbb{R}^{n\times n}$ and $\bm{X}_0 = \bm{x}_0 \bm{x}_0^{\top} \in \mathbb{R}^{n \times n}$, with $\bm{a}_i, \bm{x}_0\in \mathbb{R}^n$ for $i \in [p]$, the problem corresponds to phase retrieval from quadratic measurements of the form $y_i = \abs{\inner{\bm{a}_i, \bm{x}_0 }}^2 + \eta_i$.
\item When each $\bm{A}_i$ (for $i\in [p]$) represents an orthogonal projection and $\bm{X}_0 = \bm{x}_0 \bm{x}_0^{\top} \in \mathbb{R}^{n \times n}$, the problem becomes the fusion frame (projection) phase retrieval problem \cite{cahill2013phase,edidin2017projections}.
    \item     When $\bm{A}_i = \bm{a}_i \bm{a}_i^{\top} \in \mathbb{R}^{n\times n}$ for $i \in  [p]$ and $\bm{X}_0 \in  \mathbb{R}^{n\times n}$ is positive semidefinite, the problem reduces to the symmetric ROP ( rank-one projection
model for low-rank matrix recovery) problem \cite{ROP}.
\end{enumerate}
%Furthermore, this problem formulation extends to the realm of quantum mechanics through the framework of positive operator-valued measures (POVMs) \cite{heinosaari2013quantum}. In this context, $\{\bm{A}_i\}_{i=1}^{p}$ represents a set of positive matrices satisfying the completeness relation $\sum_{i=1}^p\bm{A}_i=\bm{I}_{n}$, where $\bm{I}_n$ denotes the identity matrix.

When the target matrix $\bm{X}_0\in \mathbb{R}^{n\times m}$  is low-rank (i.e., $\mathrm{rank}(\bm{X}_0) \leq r$), a natural approach to recover $\bm{X}_0$ is the rank-constrained least squares model
\begin{align}\label{pro:squarank}
\mathop{\argmin}_{\bm{X} \in \mathbb{R}^{n\times m}, \mathrm{rank}(\bm{X})\leq r} \ \ \ \||\mathcal{A}(\bm{X})|-\bm{y}\|_2.
\end{align}
Leveraging the relationship between matrix rank and nuclear norm in optimization, we propose a nuclear norm constrained least squares model to recover $\bm{X}_0$:
 \begin{align}\label{leastproblemnu}
\mathop{\argmin}_{\bm{X} \in \mathbb{R}^{n\times m}} \ \ \ \||\mathcal{A}(\bm{X})|-\bm{y}\|_2 \ \ \ \ \text{s.t.} \ \ \ \  \|\bm{X}\|_{*}\leq R,
\end{align}
where $\|\bm{X}\|_{*}$  denotes the nuclear norm of $\bm{X}$, and $R$ is a parameter controlling the desired rank of the solution. This formulation is analogous to the Lasso model.

In scenarios where the noise is bounded, i.e., $\|\bm{\eta}\|_2\leq \epsilon$, we can also employ the constrained nuclear norm minimization model to recover the low-rank matrix $\bm{X}_0$:
\begin{align}\label{leastproblemcon}
\mathop{\argmin}_{\bm{X} \in \mathbb{R}^{n \times m}} \ \ \ \ \|\bm{X}\|_{*} \quad \text{s.t.} \quad \||\mathcal{A}(\bm{X})| - \bm{y}\|_2 \leq \epsilon.
\end{align}
 The corresponding unconstrained approach is given by
\begin{align}\label{leastproblem}
\mathop{\argmin}_{\bm{X} \in \mathbb{R}^{n\times m}} \||\mathcal{A}(\bm{X})|-\bm{y}\|_2^2+\lambda\|\bm{X}\|_{*}.
\end{align}

\subsection{Our Contributions}

This paper aims to assess the performance of established amplitude-based models for recovering low-rank matrices. Specifically, we provide theoretical guarantees for the proposed models in the matrix recovery problem, where $\mathcal{A}: \mathbb{R}^{n \times m} \rightarrow \mathbb{R}^p$
represents a Gaussian measurement ensemble. The definition of the Gaussian measurement ensemble is given below:
\begin{definition}\label{def:G}
The linear operator $\mathcal{A}: \mathbb{R}^{n \times m} \rightarrow \mathbb{R}^p$ defined as
\[
\mathcal{A}(\bm{X}_0) = (\langle \bm{A}_1, \bm{X}_0 \rangle, \langle \bm{A}_2, \bm{X}_0 \rangle, \ldots, \langle \bm{A}_p, \bm{X}_0 \rangle)^{\top}, \text{ for } \bm{X}_0 \in \mathbb{R}^{n \times m},
\]
 is referred to as a Gaussian measurement ensemble if $\{\bm{A}_i\}_{i=1}^p$
is a set of independent random matrices, where each $\bm{A}_i \in \mathbb{R}^{n \times m}$
has entries independently drawn from the standard normal distribution $\mathcal{N}(0,1)$.
\end{definition}

%The aim of this paper is to estimate the performance of  amplitude-based models  for the low rank  matrix recovery.
Before presenting our contributions, we will first introduce some notations.
We use the notation $a\lesssim b$ to denote $a\leq C b$ for $a,b\in \mathbb{R}$, where $C$ is an absolute constant. Similarly, $a \gtrsim b$ denotes $a \geq C b$.
Additionally, $\|\cdot\|_F$ represents the Frobenius norm.

\subsubsection{ Rank-constrained Least Squares Model}
We  establish the reconstruction error  of the rank-constrained least squares model \eqref{pro:squarank} in the following theorem.
\begin{thm}\label{thm:squarank}
Let \(\mathcal{A} : \mathbb{R}^{n \times m} \rightarrow \mathbb{R}^p\) be a Gaussian measurement ensemble. Assume that \(p \gtrsim (m+n)r\). For all  matrix \(\bm{X}_0 \in \mathbb{R}^{n \times m}\) with \(\mathrm{rank}(\bm{X}_0) \leq r\), consider the measurement \(\bm{y} = |\mathcal{A}(\bm{X}_0)| + \bm{\eta}\),
where $\bm{\eta}\in {\mathbb R}^p$ is the noise vector. Let \(\hat{\bm{X}}^r \in \mathbb{R}^{m \times n}\) be any solution to \eqref{pro:squarank}, i.e.,
\begin{align*}
\hat{\bm{X}}^{r}\in \mathop{\argmin}\limits_{\bm{X} \in \mathbb{R}^{n\times m}, \mathrm{rank}(\bm{X})\leq r} \ \||\mathcal{A}(\bm{X})|-\bm{y}\|_2.
\end{align*}
Then
\begin{equation}\label{eq:upper}
\min\{\|\hat{\bm{X}}^r-\bm{X}_0\|_F, \|\hat{\bm{X}}^r+\bm{X}_0\|_F\}\lesssim \frac{\|\bm{\eta}\|_2}{\sqrt{p}}
\end{equation}
holds with probability at least $1-2\exp(-cp)$, where $c$ is a positive constant.
\end{thm}

Naturally, one may question the tightness of the upper bound presented in (\ref{eq:upper}). In Theorem \ref{thm:squarankone}, we demonstrate that this upper bound is indeed tight up to a constant factor.

\begin{thm}\label{thm:squarankone}
Let $\mathcal{A}:\mathbb{R}^{n\times m}\rightarrow \mathbb{R}^p$ be a Gaussian measurement ensemble.
 Assume that $p\gtrsim (m+n)r$ and the noise vector $\bm{\eta}=(\eta_1,\ldots,\eta_p)^\top\in {\mathbb R}^p$  satisfies
\begin{align}\label{e:noiseone}
\Big|\sum_{i=1}^p\eta_i\Big| \geq C_0\sqrt{p}\|\bm{\eta}\|_2
\end{align}
for some absolute constant $C_0\in (0,1)$.
Then for any fixed  $\bm{X}_0\in {\mathbb R}^{n\times m}$ with $\mathrm{rank}(\bm{X}_0)\leq r$,  $\|\bm{X}_0\|_F\geq \frac{\|\bm{\eta}\|_2}{\sqrt{p}}$
and $\bm{y}=|\mathcal{A}(\bm{X}_0)|+\bm{\eta}$, the solution  $\hat{\bm{X}}^{r}\in \mathbb{R}^{n\times m}$ to \eqref{pro:squarank}
 satisfies
 \begin{align}\label{eq:lower}
\min\{\|\hat{\bm{X}}^{r}-\bm{X}_0\|_F, \|\hat{\bm{X}}^{r}+\bm{X}_0\|_F\}\gtrsim \frac{\|\bm{\eta}\|_2}{\sqrt{p}}
\end{align}
 with probability at least $1-5\exp(-cp)$, where $c>0$ is an  absolute constant.
\end{thm}

  Theorem \ref{thm:squarankone} establishes that the bound $\frac{\|\bm{\eta}\|_2}{\sqrt{p}}$
 is tight for a general noise vector $\bm{\eta}$ satisfying certain conditions, including a relatively large $\abs{\sum_{i=1}^p \eta_i}$.
 Inspired by the findings presented in \cite{xia2024performance}, we establish an enhanced bound in the following theorem. This improved bound applies to noise vectors $\bm{\eta}$ that meet certain criteria, particularly when the absolute sum of their components, $\abs{\sum_{i=1}^p \eta_i}$, is relatively small. Our refinement offers a more nuanced and accurate description of how the error bound behaves under various noise conditions.

\begin{thm}\label{thm:squaranks}
Let $\mathcal{A}:\mathbb{R}^{n\times m}\rightarrow \mathbb{R}^p$ be a Gaussian measurement ensemble.
Assume that
$p\gtrsim (m+n)r$ and the noise vector $\bm{\eta}=(\eta_1,\ldots,\eta_p)^\top\in {\mathbb R}^p$  satisfies
\begin{align}\label{e:noise}
\|\bm{\eta}\|_2\lesssim \sqrt{p} \ \  \mathrm{and } \ \ \Big|{\sum_{i=1}^p\eta_i}\Big| \lesssim \sqrt{p\log p}.
\end{align}
Then,  for all $\bm{X_0}\in {\mathbb R}^{n\times m}$ with $\mathrm{rank}(\bm{X}_0)\leq r$ and  $\bm{y}=|\mathcal{A}(\bm{X}_0)|+\bm{\eta}$, the solution  $\hat{\bm{X}}^{r}\in \mathbb{R}^{n\times m}$ to \eqref{pro:squarank} satisfies
 \begin{align*}
\min\{\|\hat{\bm{X}}^{r}-\bm{X}_0\|_F, \|\hat{\bm{X}}^{r}+\bm{X}_0\|_F\}
\lesssim \max\bigg\{ \|\bm{X}_0\|_F,\sqrt{\frac{(m+n+1)r\log p}{p}}\bigg\}\sqrt{\frac{(m+n+1)r\log p}{p}}
\end{align*}
 with probability at least $1-3\exp(-cp)-4\exp(-(m+n+1)r\log p)$, where $c>0$ is an  absolute constant.
\end{thm}

\begin{remark}
Both Theorem \ref{thm:squarankone} and Theorem \ref{thm:squaranks} stipulate specific conditions that the noise must satisfy. As elucidated in \cite{xia2024performance}, a substantial class of noise vectors indeed meets these criteria. For instance, if we consider $\eta_i \in [a,b]$
  where $0<a < b$, then $\bm{\eta}$ satisfies conditions in Theorem  \ref{thm:squaranks}. Furthermore, if $\eta_i$ is independently drawn from a normal distribution $\mathcal{N}(0,\sigma^2)$, then $\bm{\eta}$ satisfies condition in Theorem \ref{thm:squaranks}  with high probability. For a more comprehensive discussion, we refer the reader to  \cite{xia2024performance}.
\end{remark}
\subsubsection{Nuclear Norm-based Models}
We present the following theorem, which establishes the performance of the nuclear norm constrained least squares model \eqref{leastproblemnu}.
\begin{thm}\label{thm:rankmin}
Let \(\mathcal{A} : \mathbb{R}^{n \times m} \rightarrow \mathbb{R}^p\) be a Gaussian measurement ensemble. Assume that \(p \gtrsim (m+n)r\). For all  matrix \(\bm{X}_0 \in \mathbb{R}^{n \times m}\) with \(\mathrm{rank}(\bm{X}_0) \leq r\), consider the measurement \(\bm{y} = |\mathcal{A}(\bm{X}_0)| + \bm{\eta}\), where $\bm{\eta}\in {\mathbb R}^p$ is the noise vector. Let \(\hat{\bm{X}}^{\ell_2} \in \mathbb{R}^{n \times m}\) be any solution to  \eqref{leastproblemnu}, i.e.,
\begin{align*}
\hat{\bm{X}}^{\ell_2}\in \mathop{\argmin}\limits_{\bm{X} \in \mathbb{R}^{n \times m}} \ \ \||\mathcal{A}(\bm{X})| - \bm{y}\|_2 \ \ \ \ \text{s.t.} \ \ \ \ \|\bm{X}\|_{*} \leq R,
\end{align*}
where   $R:=\|\bm{X}_0\|_{*}$.
Then
\[
\min\{\|\hat{\bm{X}}^{\ell_2} - \bm{X}_0\|_F, \|\hat{\bm{X}}^{\ell_2} + \bm{X}_0\|_F\} \lesssim \frac{\|\bm{\eta}\|_2}{\sqrt{p}}
\]
holds with probability at least \(1 - 2\exp(-cp)\), where \(c\) is a positive constant.
\end{thm}

Next, we present theorems that establish the performance of two nuclear norm minimization models: the constrained model \eqref{leastproblemcon} and the unconstrained model \eqref{leastproblem}.
\begin{thm}\label{thm:ell2}
Let $\mathcal{A}:\mathbb{R}^{n\times m}\rightarrow \mathbb{R}^p$ be a Gaussian measurement ensemble.
For all $\bm{X}_0\in \mathbb{R}^{n\times m}$  with $\mathrm{rank}(\bm{X}_0)\leq r$ and  noise vector $\bm{\eta}\in {\mathbb R}^p$ with $\|\bm{\eta}\|_2\leq \varepsilon$, consider the measurement
 $\bm{y}=|\mathcal{A}(\bm{X}_0)|+\bm{\eta}$.
If $p\gtrsim (m+n)r$ and  $\hat{\bm{X}}^{\ell_*}\in \mathbb{R}^{n\times m}$ is any solution to \eqref{leastproblemcon}, i.e.,
\begin{align*}
\hat{\bm{X}}^{\ell_{*}}\in\mathop{\argmin}\limits_{\bm{X} \in \mathbb{R}^{n \times m}} \ \ \|\bm{X}\|_{*} \quad \text{s.t.} \quad \||\mathcal{A}(\bm{X})| - \bm{y}\|_2 \leq \epsilon,
\end{align*}
then
$$
\min\{\|\hat{\bm{X}}^{\ell_{*}}-\bm{X}_0\|_F, \|\hat{\bm{X}}^{\ell_{*}}+\bm{X}_0\|_F\}\lesssim \frac{\varepsilon}{\sqrt{p}}
$$
holds with probability at least $1-2\exp(-cp)$, where \(c\) is a positive constant.
\end{thm}

\begin{thm}\label{thm:unell2}
Let $\mathcal{A}:\mathbb{R}^{n\times m}\rightarrow \mathbb{R}^p$ be a Gaussian measurement ensemble.
For all matrix  $\bm{X}_0\in \mathbb{R}^{n\times m}$  with $\mathrm{rank}(\bm{X}_0)\leq r$, consider the measurement
 $\bm{y}=|\mathcal{A}(\bm{X}_0)|+\bm{\eta}$, where  $\bm{\eta}\in {\mathbb R}^p$  is the noise vector.
 If $p\gtrsim (m+n)r$ and
 $\hat{\bm{X}}^{u}\in \mathbb{R}^{n\times m}$ is any solution to \eqref{leastproblem}, i.e.,
\begin{align*}
\hat{\bm{X}}^{u}\in \mathop{\arg\min}\limits_{\bm{X} \in \mathbb{R}^{n\times m}} \||\mathcal{A}(\bm{X})|-\bm{y}\|_2^2+\lambda\|\bm{X}\|_{*},
\end{align*}
where $\lambda$ is a parameter satisfying $\lambda>2\sqrt{2p}\|\bm{\eta}\|_2$,
then
$$
\min\{\|\hat{\bm{X}}^{u}-\bm{X}_0\|_F, \|\hat{\bm{X}}^{u}+\bm{X}_0\|_F\}\lesssim \frac{\lambda \sqrt{r}}{p} +\frac{\|\bm{\eta}\|_2}{\sqrt{p}}
$$
holds with probability at least $1-2\exp(-cp)$, where $c>0$ is an  absolute constant.
\end{thm}
\subsection{Related Work}
\subsubsection{Phase Retrieval}

As previously mentioned, when $m = 1$, we consider the linear operator $\mathcal{A}:{\mathbb F}^{n\times m}\to{\mathbb R}^p$ defined by
\[
 \mathcal{A}(\bm{X})=(\langle \bm{A}_1,\bm{X}\rangle,\langle \bm{A}_2,\bm{X}\rangle,\ldots,\langle \bm{A}_p,\bm{X}\rangle)^{\top}
\]
In this scenario, by setting $\bm{A}_i=\bm{a}_i\in {\mathbb F}^{n\times 1}$ and $\bm{X}_0=\bm{x}_0\in {\mathbb F}^{n\times 1}$, the amplitude-based matrix recovery problem effectively reduces to the classical vector phase retrieval problem.
 Accordingly, the rank-constrained least-squares formulation \eqref{pro:squarank} specializes to
\begin{align}\label{eq:nonlinear}
\mathop{\argmin}_{\bm{X}\in\mathbb{F}^{n\times 1} }\ \ \|\ |\mathcal{A}(\bm{X})|-\bm{y}\|_2=\mathop{\argmin}_{\bm{x}\in\mathbb{F}^{n} }\ \ \|\ |\bm{A}(\bm{x})|-\bm{y}\|_2,
\end{align}
which coincides with the nonlinear least-squares model \eqref{pro:squavector} for phase retrieval. Here, $\bm{A}:=[\bm{a}_1, \ldots, \bm{a}_p]^{\top}\in {\mathbb F}^{p\times n}$.

The performance of (\ref{eq:nonlinear}) has been analyzed in the real- and complex-valued settings in \cite{huang2020estimation} and \cite{xia2024performance}, respectively.
In the real case, when $\bm{A}=[\bm{a}_1, \ldots, \bm{a}_p]^{\top}\in \mathbb{R}^{p\times n}$ is a Gaussian matrix with entries $a_{ij}\sim \mathcal{N}(0,1)$,
it has been shown in   \cite{huang2020estimation}  that the  solution $\hat{\bm{x}}$ to \eqref{pro:squavector} (or (\ref{eq:nonlinear}))
satisfies
\begin{align}\label{errorrealbound}
\min\{\|\hat{\bm{x}}-\bm{x}_0\|_2,\|\hat{\bm{x}}+\bm{x}_0\|_2\}\lesssim \frac{\|\bm{\eta}\|_2}{\sqrt{p}}
\end{align}
with high probability, provided that $p\gtrsim n$.
This result is consistent with our findings in Theorem \ref{thm:squarank} for the case $m=1$.
Notably, it has been shown in \cite{huang2020estimation}  that the reconstruction error bound $\frac{\|\bm{\eta}\|_2}{\sqrt{p}}$
is tight. Specifically, under the conditions ${(|\sum_{i=1}^p\eta_i|)}/{p}\geq \epsilon_0$
and ${\|\bm{\eta}\|_2}/{\sqrt{p}}\leq \epsilon_1$ for some constants $\epsilon_0, \epsilon_1 > 0$, the following lower bound holds:
$$
\min\{\|\hat{\bm{x}}-\bm{x}_0\|_2,\|\hat{\bm{x}}-\bm{x}_0\|_2\}\gtrsim\frac{\|\bm{\eta}\|_2}{\sqrt{p}}.
$$
This corresponds to our result in Theorem \ref{thm:squarankone} for the case $m=1$.

%In the complex case, where $\mathbb{F}=\mathbb{C}$, the performance of the least square model \eqref{pro:squavector} is analyzed in \cite{xia2024performance}.

 \begin{comment}
  In the complex case, if $\bm{A} \in \mathbb{C}^{p \times n}$ is a complex sub-Gaussian matrix, the reconstruction error bound for \eqref{eq:nonlinear}  is $O\left(\frac{\|\bm{\eta}\|_2}{\sqrt{p}}\right)$ (see \cite{xia2024performance}).  Specifically, with high probability and provided $p\gtrsim n$,
  any solution $\hat{\bm{x}}\in\mathbb{C}^n$ to \eqref{eq:nonlinear} satisfies
\begin{align}\label{errorcombound}
\min_{\theta\in[0,2\pi)}\|\hat{\bm{x}}-\exp(\mathrm{i}\theta)\bm{x}_0\|_2\lesssim\frac{\|\bm{\eta}\|_2}{\sqrt{p}}.
 \end{align}

The error bound $\frac{\|\bm{\eta}\|_2}{\sqrt{p}}$ is also shown to be tight if $|\sum_{i=1}^p\eta_i|\geq C_0 \frac{\|\bm{\eta}\|_2}{\sqrt{p}}$, where the constant $C_0\in (0,1)$. In this scenario, the following lower bound is established:
$$
\min_{\theta\in[0,2\pi)}\|\hat{\bm{x}}-\exp(\mathrm{i}\theta)\bm{x}_0\|_2\gtrsim\frac{\|\bm{\eta}\|_2}{\sqrt{p}}.
$$
\end{comment}
When the target signal $\bm{x}_0\in \mathbb{F}^{n}$ is $k$-sparse (i.e., $\|\bm{x}_0\|_0\leq k$, where $\|\bm{x}_0\|_0$
denotes the number of nonzero entries in $\bm{x}_0$),
the sparse signal recovery problem can be formulated as the  following nonlinear  least squares model
\begin{align*}
\mathop{\argmin}_{\bm{x}\in \mathbb{F}^{n}, \|\bm{x}\|_0\leq k} \||\bm{A}\bm{x}|-\bm{\eta}\|_2^2.
\end{align*}
Two popular methods for sparse signal recovery are the constrained Lasso model
\begin{align}\label{eq:la1}
\mathop{\argmin}_{\bm{x}\in \mathbb{F}^{n}}\||\bm{A}\bm{x}|-\bm{\eta}\|_2 \ \  \ \ \text{s.t.} \ \ \ \ \|\bm{x}\|_1\leq R
 \end{align}
and its unconstrained version
\begin{align}\label{eq:la2}
\mathop{\argmin}_{\bm{x}\in \mathbb{F}^{n}} \||\bm{A}\bm{x}|-\bm{\eta}\|_2^2+\lambda \|\bm{x}\|_1,
\end{align}
where $R$ specifies the desired sparsity level of the solution, and $\lambda>0$ is the regularization parameter.

It is worth noting that for an $n \times 1$ matrix, the nuclear norm is equivalent to the $\ell_1$
  norm of the corresponding vector in $\mathbb{F}^n$
 . This establishes a natural connection between $\ell_1$
  regularization and nuclear norm regularization when $m=1$. Consequently, the constrained and unconstrained Lasso formulations (\ref{eq:la1}) and (\ref{eq:la2}) can be viewed as special cases of \eqref{leastproblemnu} and \eqref{leastproblem}, respectively, with $m=1$.

For the real case, when $\bm{A}=[\bm{a}_1, \ldots, \bm{a}_p]^{\top}\in \mathbb{R}^{p\times n}$ is a Gaussian matrix with entries $a_{ij}\sim \mathcal{N}(0,1)$, it has been established that  any solution $\hat{\bm{x}}\in \mathbb{R}^{n}$
to the constrained Lasso model (\ref{eq:la1}) with $R:= \|\bm{x}_0\|_1$
satisfies the bound in \eqref{errorrealbound} with high probability if $p\gtrsim k\log(e n/k)$. Furthermore, the solution $\hat{\bm{x}}\in \mathbb{R}^{n}$
to the unconstrained Lasso model (\ref{eq:la2}) with
$\lambda\gtrsim \|\bm{\eta}\|_1+\|\bm{\eta}\|_2\sqrt{\log n}$ and $\mathbb{F}=\mathbb{R}$ obeys
$$\min\{\|\hat{\bm{x}}-\bm{x}_0\|_2,\|\hat{\bm{x}}-\bm{x}_0\|_2\}\lesssim \frac{\lambda\sqrt{k}}{p}+ \frac{\|\bm{\eta}\|_2}{\sqrt{p}}$$
with high probability, provided $p\gtrsim k\log(e n/k)$.
For more details, refer to \cite{huang2020estimation}.
These results differ from  Theorems \ref{thm:rankmin} and \ref{thm:unell2} when $\bm{X}_0=\bm{x}_0$. This distinction arises because the estimation performance analysis in \eqref{leastproblemnu} and \eqref{leastproblem} for recovering $\bm{X}_0 = \bm{x}_0$ relies solely on the low-rank property (i.e., $\text{rank}(\bm{X}_0) = 1$) without leveraging the sparsity $k$ of the vector $\bm{x}_0$ (i.e., $\|\bm{x}_0\|_0 \leq k$).

When the noise vector $\bm{\eta}$ is bounded in the $\ell_2$
  norm by $\epsilon$ (i.e., $\|\bm{\eta}\|_2\leq \epsilon$), the constrained $\ell_1$
  minimization model can be formulated as
\begin{align}\label{eq:phasel1}
\mathop{\argmin}_{\bm{x}\in \mathbb{F}^{n}}\|\bm{x}\|_1 \ \ \ \  \text{s.t.}  \ \ \ \ \||\bm{A}\bm{x}|-\bm{\eta}\|_2\leq \epsilon.
\end{align}
 The performance of this model has been analyzed in both real and complex cases, as presented in \cite{gao2016stable} and \cite{xia2024performance}, respectively. In the real case, specifically, when \(\bm{A} = [\bm{a}_1, \ldots, \bm{a}_p]^{\top} \in \mathbb{R}^{p \times n}\) is a Gaussian matrix with entries \(a_{ij} \sim \mathcal{N}(0,1)\), it has been shown that if \(p \gtrsim k \log(e n / k)\), then with high probability, the following holds:
\begin{align}\label{errorcombound}
\min \{ \|\hat{\bm{x}}-\bm{x}_0\|_2, \|\hat{\bm{x}}+\bm{x}_0\|_2\}\,\,\lesssim\,\,\frac{\epsilon}{\sqrt{p}}
 \end{align}
for any solution \(\hat{\bm{x}} \in \mathbb{R}^n\) to the constrained \(\ell_1\) minimization model. Theorem \ref{thm:ell2} constitutes an extension of the result in (\ref{errorcombound}) to the low-rank phase retrieval setting.

\subsubsection{Low-rank Matrix Recovery}

The recovery of low-rank matrices $\bm{X}_0\in \mathbb{F}^{n\times m}$ from linear measurements $\bm{y} = \mathcal{A}(\bm{X}_0) + \bm{\eta}$, has been a subject of intensive research in recent years.
 This problem, often referred to as affine rank minimization \cite{recht2010guaranteed}, has strong theoretical and algorithmic connections to compressed sensing \cite{foucart2013invitation}.

A widely adopted model for recovering the low-rank matrix $\bm{X}_0$
  with $\mathrm{rank}(\bm{X}_0)\leq r$ is formulated as
\begin{align*}%\label{pro:rankcon}
\mathop{\argmin}_{\bm{X}\in \mathbb{F}^{n\times m}, \mathrm{rank}(\bm{X})\leq r } \|\mathcal{A}(\bm{X})-\bm{y}\|_2^2.
\end{align*}
%Various algorithms for solving this minimization problem  have been introduced and examined in the literature, including normalized iterative hard thresholding \cite{tanner2013normalized}, Procrustes flow \cite{tu2016low}, GNMR \cite{zilber2022gnmr}, scaled gradient descent \cite{tong2021accelerating}, and recursive importance sketching \cite{luo2024recursive}.
In scenarios where the noise is bounded, i.e., $\|\bm{\eta}\|_2\leq \epsilon$, an alternative model of particular theoretical and practical interest for recovering $\bm{X}_0$
  is given by
\begin{align*}%\label{pro:nuclearnorm}
\mathop{\argmin}_{\bm{X}\in \mathbb{F}^{n\times m} } \|\bm{X}\|_{*} \ \ \ \ \text{s.t.} \ \ \ \ \|\mathcal{A}(\bm{X})-\bm{y}\|_2\leq \epsilon.
\end{align*}
For this constrained  problem with $\mathbb{F}=\mathbb{R}$, a variety of sufficient conditions based on the restricted isometry property (RIP) for matrices have been established to guarantee the stable recovery of all matrices with rank at most $r$. Comprehensive analyses of these conditions are elucidated in \cite{cai2013sparse, candes2011tight, mo2011new, wang2013bounds}.
Additionally, the following unconstrained minimization problem has been extensively studied:
\begin{align*}
\mathop{\argmin}_{\bm{X} \in \mathbb{F}^{n\times m} }  \|\mathcal{A}(\bm{X})-\bm{y}\|_2^2+\lambda\|\bm{X}\|_{*}.
\end{align*}
%Substantial advancements in addressing this unconstrained optimization problem have been achieved through a series of seminal works, e.g. \cite{goldfarb2011convergence, lai2013improved, toh2010accelerated}.
In the real case (i.e., $\mathbb{F}=\mathbb{R}$),
when $\mathcal{A}$ is a  Gaussian measurement ensemble, it has been shown in \cite{candes2011tight} that
the solution $\hat{\bm{X}}\in \mathbb{R}^{n\times m}$
to the unconstrained minimization problem
satisfies
\[
\|\hat{\bm{X}} - \bm{X}_0\|_F \lesssim \frac{{\lambda\sqrt{r}}}{p}
\]
 with high probability, where the parameter $\lambda=16\sqrt{\max\{m,n\}}\cdot \sigma$ and the noise
 $\bm{\eta}\in \mathcal{N}(0,\sigma^2\bm{I}_p)$,
 provided $p \gtrsim (m+n+1)r$. Comparing this error bound with that in Theorem \ref{thm:unell2}, we see that although both approaches provide recovery guarantees, the nonlinear Lasso \eqref{leastproblem} achieves a weaker bound for amplitude-based matrix recovery. This is expected, as it has access only to the magnitudes of the measurements (the phase information is lost) and imposes no assumptions on the noise.

\begin{comment}
When $\bm{A}_i\in \mathbb{F}^{n\times m}$ with $m=n$ is a positive semidefinite Hermitian matrix, and $\bm{X}_0=\bm{x}_0\bm{x}_{0}^{*}$ with
 $\bm{x}_{0}\in \mathbb{F}^{n}$,
it is useful to observe  that
\[
\langle \bm{A}_i, \bm{X}_0\rangle=\mathrm{Tr}(\bm{A}_i\bm{x}_0\bm{x}_{0}^{*})=\mathrm{Tr}(\bm{x}_{0}^{*}\bm{A}_i\bm{x}_0)=
\bm{x}_{0}^{*}\bm{A}_i\bm{x}_0\geq0.
\]
 This  non-negativity property leads to $|\langle \bm{A}_i, \bm{X}_0\rangle|=\langle \bm{A}_i, \bm{X}_0\rangle.$
Consequently, the amplitude-based matrix recovery problem can be reformulated as the recovery of the rank-one Hermitian matrix
 $\bm{x}_0\bm{x}_{0}^{*}$ from
\[
(\langle \bm{A}_1, \bm{X}_0\rangle, \langle \bm{A}_2, \bm{X}_0\rangle, \ldots,\langle \bm{A}_p, \bm{X}_0\rangle)^{\top}.
\]
\end{comment}

We next turn to the rank-one measurements.
Let $\bm{A}_i = \bm{a}_i \bm{a}_i^{*}$ for some $\bm{a}_i \in \mathbb{F}^n$, and let $\bm{X}_0 \in \mathbb{F}^{n \times n}$ be a positive semidefinite Hermitian matrix with rank at most $r$. We then observe that
$$
\langle \bm{A}_i, \bm{X}_0\rangle=\mathrm{Tr}(\bm{a}_i \bm{a}_i^{*}\bm{X}_0)
=\mathrm{Tr}(\bm{a}_{i}^{*}\bm{X}_0\bm{a}_i)=
\bm{a}_{i}^{*}\bm{X}_0\bm{a}_i\geq 0,
$$
which implies $|\langle \bm{A}_i, \bm{X}_0\rangle|=\langle \bm{A}_i, \bm{X}_0\rangle.$ Hence,
 the amplitude-based matrix recovery problem can be reformulated as the recovery of the rank-$r$ Hermitian matrix $\bm{X}_0$ from
 \[
 \mathcal{A}(\bm{X}_0)=(\langle \bm{A}_1, \bm{X}_0 \rangle, \langle \bm{A}_2, \bm{X}_0 \rangle, \ldots, \langle \bm{A}_p, \bm{X}_0 \rangle)^{\top}.
 \]
In this setting, Kueng, Rauhut, and Terstiege \cite{kueng2017low} analyze the performance of the convex program:
\begin{align*}%\label{pro:nuclearnorm}
\mathop{\argmin}_{\bm{X}\in \mathbb{F}^{n\times n}, \bm{X} \succeq \mathbf{0} } \|\bm{X}\|_{*} \ \ \ \ \text{s.t.} \ \ \ \ \|\mathcal{A}(\bm{X})-\bm{y}\|_2\leq \epsilon.
\end{align*}
 The authors in \cite{kueng2017low} demonstrate  that the solution $\hat{\bm{X}}$ to this model with $\mathbb{F}=\mathbb{C}$ satisfies
\[
\|\hat{\bm{X}} - \bm{X}_0\|_F \lesssim \frac{{\epsilon}}{\sqrt{p}}
\]
with high probability, provided  $p \gtrsim nr$ and $\bm{a}_1, \ldots, \bm{a}_p \in \mathbb{C}^{n}$ are independent standard Gaussian  random vectors.
This error bound is analogous to the result presented in Theorem  \ref{thm:ell2} with $m=n$.
%Furthermore, comparable estimation performance is achieved when $p \gtrsim nr \log n$, given that $\bm{a}_1, \ldots, \bm{a}_p$ are independently drawn from a complex projective $t$-design.
A notable special case of this framework is the recovery of the Hermitian matrix $\bm{X}_0 = \bm{x}_0 \bm{x}_0^{*}$, which reduces to phase retrieval based on $y_i = |\langle \bm{a}_i, \bm{x}_0 \rangle|^2 + \eta_i$ for $i \in [p]$. This corresponds to the PhaseLift approach for phase retrieval, which is a topic extensively explored e.g., \cite{candes2014solving, candes2013phaselift}.

\subsubsection{Low-rank Phase Retrieval}

 Vaswani, Nayer and Eldar \cite{vaswani2017low} introduce a seminal magnitude-only (phaseless) measurement model for LRPR, defined as
\[
y_{ik} = |\langle \bm{a}_{i,k}, \bm{x}_k^{0}\rangle|, \quad i \in [p], \quad k \in [m],
\]
where \( \bm{x}_k^{0}\in \mathbb{R}^n \) denotes the \( k \)-th column of \( \bm{X}_0 \in \mathbb{R}^{n\times m}\), and the measurement vectors \( \bm{a}_{i,k}\in \mathbb{R}^n \) are mutually independent.
The authors  in \cite{vaswani2017low} propose two iterative algorithms comprising a spectral initialization step followed by an iterative procedure to maximize the likelihood of the observed data. Furthermore, they establish sample complexity bounds for the initialization approach to ensure a reliable approximation of the true matrix \( \bm{X}_0 \).
Subsequent to this foundational work, the LRPR problem based on this measurement model has attracted considerable scholarly attention e.g., \cite{9897747, pmlr-v97-nayer19a, 9050801,9537806,9912351}. Notably, Nayer and Vaswani \cite{9050801} propose an alternating minimization algorithm for LRPR (AltMinLowRaP) that exhibits geometric convergence. They demonstrate that the AltMinLowRaP estimate converges geometrically to \( \bm{X}_0 \) under the condition that \( mp\geq Cnr^4 \log(1/\varepsilon) \), where \( \bm{a}_{i,k} \) are independent and identically distributed standard Gaussian vectors, and the right singular vectors of \( \bm{X}_0 \) satisfy the incoherence assumption. In their follow-up work \cite{9537806}, an enhanced convergence guarantee for AltMinLowRaP is presented. Specifically, they prove that if the right singular vectors of \( \bm{X}_0 \) satisfy the incoherence assumption, the AltMinLowRaP estimate converges geometrically to \( \bm{X}_0 \) when the total number of measurements satisfies \( m p \geq n r^2 (r + \log(1/\varepsilon)) \).

An alternative phaseless linear measurement model for LRPR in \cite{lee2021phase} is introduced:
\begin{align*}
\bm{y}=|\mathcal {A}(\bm{X}_0)|^2+\bm{\eta},
\end{align*}
which is different from our amplitude-based measurements $\bm{y}=|\mathcal{A}(\bm{X}_0)|+\bm{\eta}$.
  The authors in \cite{lee2021phase} develop recovery algorithms utilizing the concept of anchored regression \cite{bahmani2017phase, bahmani2019solving} to reconstruct \(\bm{X}_0 \in \mathbb{F}^{n\times m}\) from \(p\) quadratic measurements. These proposed algorithms demonstrate the capability to recover the low-rank \(\bm{X}_0\) from phaseless measurements using significantly fewer than \(mn\) measurements.
Specifically, when the matrices
$\bm{A}_i=\bm{a}_i\bm{b}_i^{*}$ with the independent vectors
 $\bm{a}_i \sim \mathcal{N}(0,\bm{I}_n)$ and $\bm{b}_i \sim \mathcal{N}(0,\bm{I}_m)$
 represent rank-$1$ complex Gaussian matrices,  the anchored regression returns an accurate estimate of the  complex rank-$1$ matrix \( \bm{X}_0 \) with high probability,
  provided that $p/\ln^2 p\gtrsim m+n $.
For a rank $r$ matrix \( \bm{X}_0 \in \mathbb{R}^{n\times m} \),
if  the matrices   $\bm{A}_i\in \mathbb{R}^{n\times m} $  are  Gaussian random matrices,
the anchored regression provides an accurate estimate of the  complex rank-$1$ matrix \( \bm{X}_0 \) with high probability,
 given that $p\gtrsim r(m+n)\ln(m+n)$.

Whereas most work on low-rank phase retrieval (LRPR) has focused on developing and analyzing theoretical algorithms, our study examines the performance of established amplitude-based formulations for low-rank matrix recovery. This complements the algorithmic literature and contributes to a more comprehensive understanding of LRPR and its practical applications.

\section {Preliminaries}

\subsection{Covering Numbers of Sets}
To set the stage for our analysis, we begin by defining the following sets of matrices. Let
\begin{align}\label{setur}
U_r:=\{\bm{H}\in \mathbb{R}^{n\times m}:\|\bm{H}\|_F=1, \mathrm{rank}(\bm{H})\leq r\};\\
\label{setkr}
K_r:=\{\bm{H}\in \mathbb{R}^{n\times m}:\|\bm{H}\|_F\leq 1, \mathrm{rank}(\bm{H})\leq r\};\\
\label{setnr}
N_r^{*}:=\{\bm{H}\in \mathbb{R}^{n\times m}:\|\bm{H}\|_F\leq 1, \|\bm{H}\|_{*}\leq \sqrt{r}\}.
\end{align}
It is evident that $U_r\subseteq K_r\subseteq N_r^{*}$, given that $\|\bm{H}\|_{*}\leq \sqrt{\mathrm{rank}(\bm{H})}\cdot \|\bm{H}\|_{F}$.
In what follows, we establish bounds on the covering numbers of these sets.
\begin{lemma}\cite[Lemma 3.1]{candes2011tight}\label{cover}
For  the set $U_r$,
there exists an $\epsilon$-net {$\overline{U_r} \subset U_r$} with respect to the Frobenius norm obeying
$$
\#\overline{U_r}\leq \left(\frac{9}{\epsilon}\right)^{(m+n+1)r}.
$$
\end{lemma}

\begin{remark}\label{rmk:cover}
The proof of Lemma \ref{cover} hinges on a fundamental property of $\varepsilon$-nets
as elucidated in \cite[III.1]{candes2011tight}. Let $S$ denote a unit ball in an $n$-dimensional space with respect to a given norm $\|\cdot\|_F$, or alternatively, the surface of the unit ball or any other subset of the unit ball. In such cases, the cardinality of the $\epsilon$-net $\overline{S}$ of $S$ is bounded above as follows
\begin{align}\label{e:unitball}
\#\overline{S}\leq \left(1+\frac{2}{\epsilon}\right)^n \leq \left(\frac{3}{\epsilon}\right)^n.
\end{align}
The latter inequality holds universally under the assumption that $0<\epsilon \leq 1$.
In the proof of Lemma \ref{cover}, the upper bound $\#\overline{S} \leq \left(\frac{3}{\epsilon}\right)^n$ is utilized.
Importantly, the proof of Lemma \ref{cover} remains valid when utilizing the tighter bound
$\#\overline{S} \leq \left(1+\frac{2}{\epsilon}\right)^n$. Employing this bound, we derive the following  inequalities
\begin{align*}
\#\overline{U_r} \leq \left(1+\frac{6}{\epsilon}\right)^{(m+n+1)r}. %\ \ \ \ \ \  \#\overline{K}_r \leq (1+6/\epsilon)^{(m+n+1)r},
\end{align*}

%where $\overline{K}_r$ is  an $\epsilon$-net of the set ${K}_r$ with respect to the Frobenius norm.
%The bound on the cardinality of $\#\overline{K}_r$ plays a pivotal role in establishing
%the covering number of the set $N_r^{*}$ in the subsequent lemma.
\end{remark}

\begin{lemma}\label{add:cover}%\cite[Lemma 3.1]{candes2011tight}\label{cover}
For  the set ${K}_r$,
there exists an $\epsilon$-net {$\overline{K_r}  \subset K_r$} with respect to the Frobenius norm obeying
$$\#\overline{K_r}\leq \left(1+\frac{6}{\epsilon}\right)^{(m+n+1)r}.$$
\end{lemma}
Although the proof of Lemma \ref{add:cover} closely parallels that of Lemma \ref{cover} in \cite{candes2011tight},
we present a detailed exposition here for the sake of clarity and completeness.
\begin{proof}
Assume that the singular value decomposition of $\bm{X}\in K_r$ is
$\bm{X}=\bm{U}\bm{\Sigma}\bm{V}^{T}$,
where $\bm{\Sigma}\in \mathbb{R}^{r \times r}$ is a diagonal matrix with non-negative singular values  satisfying $\|\bm{\Sigma}\|_F\leq 1$, and $\bm{U} \in \mathbb{R}^{n \times r}$
  and $\bm{V} \in \mathbb{R}^{m \times r}$  are matrices with orthonormal columns comprising the left- and right-singular vectors, respectively. We shall construct  an $\epsilon$-net $\overline{K_r}$ for $K_r$ by separately  covering the set of permissible $\bm{U}$, $\bm{V}$ and $\bm{\Sigma}$.
Let
\[
D := \left\{ \bm{\Lambda} \in \mathbb{R}^{r \times r}: \bm{\Lambda} \text{ is diagonal}, \ \Lambda_{ii} \geq 0 \ \text{for all} \  i, \ \|\bm{\Lambda}\|_F \leq 1 \right\}.
\]
  Evidently, $\bm{\Sigma}\in D$ and
$D$ is a  subset of the unit ball under the norm $\|\cdot\|_F$. By \eqref{e:unitball}, there exists
an $\frac{\epsilon}{3}$-net  $\overline{D}$ for $D$ with cardinality $\#\overline{{D}}\leq \left(1+\frac{6}{\epsilon}\right)^r$.
Consequently, there exists  $\bar{\bm{\Sigma}}\in \overline{D}$ such that $\|\bm{\Sigma}-\bar{\bm{\Sigma}}\|_F\leq \frac{\epsilon}{3}$.
Set
\[
O_n:=\{\bm{Q}\in \mathbb{R}^{n\times r}: \ \bm{Q}^{\top}\bm{Q}=\bm{I}_r\}
\]
 and note that
$\bm{U}\in O_n$. To cover $O_n$, we introduce  the norm $\|\cdot\|_{1,2}$ defined as
$\|\bm{Q}\|_{1,2}=\max_{i\in[r]}\|\bm{Q}_i\|_2$, where $\bm{Q}_i$ denotes the $i$-th column of $\bm{Q}$.
 Observe that $O_n$ is a subset of the unit ball under the norm $\|\cdot\|_{1,2}$, as each column of an orthogonal matrix has unit $\ell_2$ norm.
Thus, by \eqref{e:unitball}, there exists an $\frac{\epsilon}{3}$-net $\overline{O_n}$ for $O_n$ with cardinality  $\#\overline{O_n}\leq \left(1+\frac{6}{\epsilon}\right)^{nr}$.  Hence, there exists $\bar{\bm{U}}\in \overline{O_n}$ such that $\|\bm{U}-\bar{\bm{U}}\|_{1,2}\leq \frac{\epsilon}{3}$.

 Similarly, for the set $O_m=\{\bm{R}\in \mathbb{R}^{m\times r}:\  \bm{R}^{\top}\bm{R}=\bm{I}_r\}$, there exists
 an $\frac{\epsilon}{3}$-net $\overline{O_m}$  with cardinality $\#\overline{O_m}\leq \left(1+\frac{6}{\epsilon}\right)^{mr}$.
  Hence, there exists  $\bar{\bm{V}}\in \overline{O_m}$ such that $\|\bm{V}-\bar{\bm{V}}\|_{1,2}\leq \frac{\epsilon}{3}$.

Set
\[
\overline{K_{r}}\,\,:=\,\,\{\bar{\bm{Q}}\bar{\bm{\Sigma}}\bar{\bm{R}}^{\top}:\
\bar{\bm{Q}}\in \overline{O_n}, \bar{\bm{\Sigma}}\in \overline{D}, \bar{\bm{R}}\in \overline{O_m}\}.
\]
 Then
$\#\overline{K_{r}}\leq (\#\overline{O_n})\cdot (\#\overline{O_m})\cdot(\#\overline{D})\leq \left(1+\frac{6}{\epsilon}\right)^{(m+n+1)r}$,
and, in particular,   $\bar{\bm{X}}=\bar{\bm{U}}\bar{\bm{\Sigma}}\bar{\bm{V}}^{\top}\in \overline{K_{r}}$.
Furthermore, by the triangle inequality, we have
\begin{align*}
\|\bm{X}-\bar{\bm{X}}\|_F
&=\|\bm{U}\bm{\Sigma}\bm{V}^{\top}-\bar{\bm{U}}\bar{\bm{\Sigma}}\bar{\bm{V}}^{\top}\|_F\\
&=\|(\bm{U}-\bar{\bm{U}})\bm{\Sigma}\bm{V}^{\top}+\bar{\bm{U}}\bm{\Sigma}\bm{V}^{\top}
+\bar{\bm{U}}\bar{\bm{\Sigma}}(\bm{V}^{\top}-\bar{\bm{V}}^{\top})-\bar{\bm{U}}\bar{\bm{\Sigma}}\bm{V}^{\top}\|_F\\
&\leq \|(\bm{U}-\bar{\bm{U}})\bm{\Sigma}\bm{V}^{\top}\|_F
+\|\bar{\bm{U}}(\bm{\Sigma}-\bar{\bm{\Sigma}})\bm{V}^{\top}\|_F
+\|\bar{\bm{U}}\bar{\bm{\Sigma}}(\bm{V}^{\top}-\bar{\bm{V}}^{\top})\|_F.
\end{align*}
We bound the three terms separately.
For the first term, using the orthogonality of $\bm{V}$, we obtain
$$\|(\bm{U}-\bar{\bm{U}})\bm{\Sigma}\bm{V}^{\top}\|_F^2=\|(\bm{U}-\bar{\bm{U}})\bm{\Sigma}\|_F^2
=\sum_{i\in[r]}\Sigma_{ii}^2\|\bm{U}_i-\bar{\bm{U}}_i\|_2^2\leq \|\bm{\Sigma}\|_F^2\|\bm{U}-\bar{\bm{U}}\|_{1,2}^2\leq \left(\frac{\epsilon}{3}\right)^2.$$
Similarly, for the last term, we have $\|\bar{\bm{U}}\bar{\bm{\Sigma}}(\bm{V}^{\top}-\bar{\bm{V}}^{\top})\|_F\leq \frac{\epsilon}{3}$.
For the middle term, we observe that
\[
\|\bar{\bm{U}}(\bm{\Sigma}-\bar{\bm{\Sigma}})\bm{V}^{\top}\|_F=\|\bm{\Sigma}-\bar{\bm{\Sigma}}\|_F\leq \frac{\epsilon}{3},
\]
exploiting the orthogonality of $\bar{\bm{U}}$ and $\bm{V}$.

Consequently, for any $\bm{X}\in K_r$ there exists $\bar{\bm{X}}\in \overline{K}_r$ such that $\|\bm{X}-\bar{\bm{X}}\|_F\le \epsilon$. Thus, $\overline{K}_r$ is an $\epsilon$-net for $K_r$, which completes the proof.

\end{proof}

\begin{lemma}\label{cover2}
For the set $N_r^{*}$ , there exists an $\epsilon$-net $\overline{N_r^{*}} \subset N_r^{*}$
with respect to the Frobenius norm, whose cardinality satisfies
$$\#\overline{N_r^{*}}\leq \exp(24(m+n+1)r/\epsilon^3).$$
\end{lemma}
\begin{proof}

For any matrix $\bm{X}\in N_r^{*}$, let $(\sigma_1,\ldots,\sigma_n)$ denote its vector of singular values, and let $\bm{X}_t$ be a best rank-$t$ approximation to $\bm{X}$ in the Frobenius norm (i.e., $\bm{X}_t \in \arg\min_{\mathrm{rank}(\bm{Z})\le t}\|\bm{X}-\bm{Z}\|_F$), where $t$ will be specified below.
We have
\begin{align*}
\|\bm{X}-\bm{X}_t\|_F\,\,=\,\,\sqrt{\sum_{i\geq t+1}\sigma_i^2}\,\,\overset{(a)}\leq\,\, \frac{\sum_{i\geq 1}\sigma_i}{2\sqrt{t}}\,\,=\,\,
\frac{\|\bm{X}\|_{*}}{2\sqrt{t}}.
\end{align*}
Here, the inequality ($a$) follows from $s_t(\bm{x})_2\leq \frac{1}{2\sqrt{t}}\|\bm{x}\|_1$,
 where $s_t(\bm{x})_2 $  is the $\ell_2$-error of the best $t$-term approximation to $\bm{x}\in \mathbb{R}^n$ (see \cite[Theorem 2.5]{foucart2013invitation}).
Note that \( \bm{X}_t \in K_t \). To establish
the  covering number estimate of  $N_r^{*}$, we first consider an \( \epsilon/2 \)-net \( \overline{K_t} \) for  \( K_t \).
According to Lemma  \ref{add:cover}, we have  $\# \overline{K_t} \leq (1+6/(\epsilon/2))^{(m+n+1)t}$.
Consequently, there exists \( \bm{X}_t^{\epsilon} \in \overline{K_t} \) such that \( \|\bm{X}_t - \bm{X}_t^{\epsilon}\|_F \leq \frac{\epsilon}{2} \). Thus, we obtain
\begin{align*}
\|\bm{X}-\bm{X}_t^{\epsilon}\|_F\leq \|\bm{X}-\bm{X}_t\|_F+\|\bm{X}_t-\bm{X}_t^{\epsilon}\|_F
\leq \frac{\|\bm{X}\|_{*}}{2\sqrt{t}}+\frac{\epsilon}{2}
\leq\frac{\sqrt{r}}{2\sqrt{t}}+\frac{\epsilon}{2},
\end{align*}
where the last inequality is from the fact $\bm{X} \in N_r^{*}$, i.e., $\|\bm{X}\|_*\leq \sqrt{r}$.
By choosing $t=\lceil r/\epsilon^2\rceil$ (throughout the paper we assume $\epsilon\le \sqrt{r}$), we have $r/\epsilon^2 \le t \le 2r/\epsilon^2$, which ensures that $\|\bm{X}-\bm{X}_t^{\epsilon}\|_F \le \epsilon$.
Therefore, $\overline{K_t}$ serves as an  $\epsilon$-net of $N_r^{*}$ with
\[
\# \overline{K_t} \leq (1+6/(\epsilon/2))^{(m+n+1)t}\leq \exp(24(m+n+1)r/\epsilon^3).
\]
\end{proof}

\begin{comment}

\begin{remark}
In the vector case, a comparable result can be observed: the covering number of the set \( V_{s} := \{\bm{x} \in \mathbb{R}^n : \|\bm{x}\|_2 \leq 1, \|\bm{x}\|_1 \leq \sqrt{s}\} \) is bounded above by \(\exp(Cs \log(en/s) / \epsilon^2)\) \cite[Lemma 3.4]{plan2013one}.
For the matrix set \( M_{r} := \{\bm{X} \in \mathbb{R}^{n \times n} : \|\bm{X}\|_F \leq 1, \frac{\|\bm{X}\|_{*}}{\|\bm{X}\|_F} \leq r\} \), its covering number is bounded above by \(\exp(72nr / \epsilon^3)\) \cite[(22)]{foucart2019recovering}.
\end{remark}
\end{comment}
\subsection{Gaussian-type Concentration Inequalities}

We recall the following Gaussian concentration inequality and two auxiliary lemmas, which will be used in the proof of Lemma \ref{spaceSRIP}.

\textbf{Gaussian Concentration Inequality  \cite{recht2010guaranteed}:} Let  $\mathcal{A}:{\mathbb R}^{n\times m}\to{\mathbb R}^p$ be  a Gaussian measurement ensemble. For any given \(\bm{X} \in \mathbb{R}^{ n\times m}\) and any fixed \(0 < t < 1\), there is
\begin{align}\label{e:ceninequality}
\mathbb{P}\Bigg[\Big|\frac{1}{p}\|\mathcal{A}(\bm{X})\|_2^2 - \|\bm{X}\|_F^2\Big| \geq t\|\bm{X}\|_F^2\Bigg] \leq 2\exp\Big(-\frac{p}{2}\Big(\frac{t^2}{2} - \frac{t^3}{3}\Big)\Big).
\end{align}
%The following two lemmas are instrumental in our analysis.
\begin{lemma}\cite[Lemma 4.2]{voroninski2016strong}\label{lem:xu1}
Assume that $y_1,\ldots,y_p$ are i.i.d $\mathcal{N}(0,1)$ and set
$$
\mu_p:=\frac{1}{\sqrt{p}}\mathbb{E}\left[\sqrt{\sum_{i=1}^{\lceil p/2\rceil}|y|_{(i)}^2}\right],
$$
where $|y|_{(1)}\leq |y|_{(2)}\leq \cdots \leq |y|_{(p)}$ (i.e., $\{|y|_{(i)}\}_{i=1}^p$ is a rearrangement of $\{|y_{i}|\}_{i=1}^p$).
Then for any $p\geq 1$, there is
$$\mu_{p}\geq\nu_0,$$
where $\nu_0:=\frac{1}{18}\sqrt{\frac{\pi}{2}}\approx0.0696$.
\end{lemma}
\begin{lemma}\cite[Lemma 4.3]{voroninski2016strong}\label{lem:xu2}
Under the conditions of Lemma \ref{lem:xu1}, there is
\begin{align*}
\mathbb{P}\left[(\mu_p-t)^2\leq \frac{1}{p}\sum_{i=1}^{\lceil p/2\rceil}|y|_{(i)}^2\leq(\mu_p+t)^2\right]\geq1-2\exp(-pt^2/2),
\end{align*}
where $0\leq t\leq \mu_p$.
\end{lemma}

\subsection{Strong Restricted Isometry Property (SRIP) for Matrices}

We begin by revisiting the definition of the restricted isometry property (RIP) for linear maps. Let ${\mathcal A}:{\mathbb R}^{n\times m}\rightarrow {\mathbb R}^p$ denote a linear map. We say that ${\mathcal A}$ satisfies the RIP of order $r$, where  $1 \leq r \leq \min\{m, n\}$, with isometry constant $\delta_r\in[0,1)$ if the following condition holds for all matrices $\bm{X}\in \mathbb{R}^{n\times m}$ of rank at most $r$:
\begin{align}\label{e:matrixRIP}
(1-\delta_{r})\|\bm{X}\|_F^2\leq \frac{1}{p}\|\mathcal{A}(\bm{X})\|_2^2\leq (1+\delta_{r})\|\bm{X}\|_F^2.
\end{align}
It is well-established in the extant literature that  the Gaussian measurement ensemble $\mathcal{A}$  satisfies the RIP of order $r$ with high probability, provided that $p\gtrsim(m+n+1)r$. For a detailed discussion on this topic, the reader is referred to \cite{candes2011tight} and \cite{recht2010guaranteed}.

We now define the strong restricted isometry property (SRIP) for matrices. This definition extends the SRIP concept originally introduced for vectors by Voroninski and Xu \cite{voroninski2016strong}.
\begin{definition}\label{def:SRIPM}
 Suppose  $\mathcal{A}:\mathbb{R}^{n\times m}\rightarrow \mathbb{R}^p$ is a  linear map, and
   $1 \leq r \leq \min\{m,n\}$ is an integer. $\mathcal{A} $ satisfies SRIP
   % the strong restricted isometry property
   of  order $r$  and levels $\theta_{-}, \theta_{+} \in (0,2)$ if
  \begin{equation}\label{eq:sripma}
      \theta_{-}\|\bm{X}\|_F^2 \leq \min\limits_{I\subseteq [p],\#I\geq p/2}\frac{1}{p} \|\mathcal{A}_{I} (\bm{X})\|_2^2 \leq \max\limits_{I\subseteq [p],\#I\geq p/2}\frac{1}{p} \|\mathcal{A}_{I} (\bm{X})\|_2^2\leq
      \theta_{+}\|\bm{X}\|_F^2
  \end{equation}
holds  for all  matrices $\bm{X}$ with rank at most $r$.
Here, $\mathcal{A}_{I}:{\mathbb R}^{n\times m}\rightarrow {\mathbb R}^{\# I}$ is defined as
\[
{\mathcal A}_{I}(\bm{X})=(\langle \bm{A}_{i_1}, \bm{X}_0\rangle, \langle \bm{A}_{i_2}, \bm{X}_0\rangle,\ldots,\langle \bm{A}_{i_{\#I}}, \bm{X}_0\rangle)^{\top},
\]
where $ I:=\{i_1,\ldots,i_{\#I}\}\subset [p]$.
\end{definition}
\begin{remark}
For an \( n \times m \) Gaussian matrix \( \bm{A} \) with i.i.d \( \mathcal{N}(0,1) \) entries, it is well-established in \cite{baraniuk2008simple} and \cite{voroninski2016strong} that $\bm{A}$ satisfies both  RIP and SRIP of order \( k \leq m \) with high probability when \( n\geq Ck \log(em/k) \), where  \( k \in \mathbb{Z}_+ \), and \( C \) is an absolute constant.
% Extending these concepts to the matrix setting,  let \( \mathcal{A} : \mathbb{R}^{n \times m} \rightarrow \mathbb{R}^p \) be a Gaussian measurement ensemble.  It has been shown that
 % \( \mathcal{A} \) satisfies the RIP of order \( 1 \leq r \leq \min\{m, n\} \) with high probability, provided that \( p \geq C(m+n)r \), where \( C \) is an absolute constant.
 %This result, which establishes the RIP for linear maps acting on low-rank matrices,  has been comprehensively analyzed in  \cite{candes2011tight} and \cite{recht2010guaranteed}.
% Building upon these foundational results,
\end{remark}
Next, we  seek to establish that the SRIP,  with constants \( \theta_{-}, \theta_{+} \in (0,2) \) also holds for the Gaussian measurement ensemble \( \mathcal{A} \) with high probability.
\begin{lemma} \label{spaceSRIP}
Let $\mathcal{A}:\mathbb{R}^{n\times m}\rightarrow \mathbb{R}^p$ be a Gaussian measurement ensemble and
$r\in \mathbb{Z}_{+}$ with $1\leq r\leq \min\{m,n\}$. Assume that
$p\geq C(m+n)r$, where $C>0$ is an absolute constant. Then there exist constants $0<\theta_{-}\leq \theta_+<2$  such that $\mathcal{A}$ satisfies SRIP of order $r$ with  probability
at least $1-\exp(-cp)$, where $c>0$  is also an absolute constant.
\end{lemma}

We now prove Lemma \ref{spaceSRIP} by leveraging the results established in Lemmas \ref{cover}, \ref{lem:xu1}, and \ref{lem:xu2}. The proof proceeds in two steps: first, we derive probabilistic bounds ensuring that (\ref{eq:sripma}) holds for a fixed matrix; second, we extend the result uniformly to all matrices of rank at most $r$ via an $\epsilon$-net argument.

\begin{proof}
Step $1$: Probabilistic bounds for a fixed matrix.

Let $\bm{X} \in \mathbb{R}^{n\times m}$
  be fixed. We aim to show that there exist constants $c_{-},c_{+} \in (0,2)$ with $c_{-}\leq c_{+}$ such that
\begin{align}\label{e:SRIPFIXX}
      c_{-}\|\bm{X}\|_F^2 \leq \min\limits_{I\subseteq [p],\#I\geq p/2}\frac{1}{p}\|\mathcal{A}_{I} (\bm{X})\|_2^2
      \leq \max\limits_{I\subseteq [p],\#I\geq p/2}\frac{1}{p}\|\mathcal{A}_{I} (\bm{X})\|_2^2\leq
      c_{+}\|\bm{X}\|_F^2
\end{align}
holds with probability at least $1-4\exp(-cp)$, where $c$ is an absolute constant.
Without loss of generality, assume \( \|\bm{X}\|_F = 1 \). Define \( \bm{y} := \mathcal{A}(\bm{X}) \) and \( \bm{y}_I := \mathcal{A}_I(\bm{X}) \) for \( I \subseteq [p] \). Since each \( \bm{A}_i \in \mathbb{R}^{n \times m} \) for \( i \in [p] \) is a Gaussian matrix with entries \( a_{jk} \sim \mathcal{N}(0, 1) \) and $\|\bm{X}\|_F=1$, we have  \( y_i = \langle \bm{A}_i, \bm{X} \rangle \sim \mathcal{N}(0, 1) \).

Lower bound for \eqref{e:SRIPFIXX}: For independent standard normal variables $y_1, \ldots, y_p$, using  Lemma \ref{lem:xu1}, we have
$$\mu_p=\frac{1}{\sqrt{p}}\mathbb{E}\left[\sqrt{\sum_{i=1}^{\lceil p/2\rceil}|y|_{(i)}^2}\right]\geq\nu_0$$
 for $\nu_0=\frac{1}{18}\sqrt{\frac{\pi}{2}}$, where $\{|y|_{(i)}\}_{i=1}^p$ satisfying $|y|_{(1)}\leq |y|_{(2)}\leq \ldots \leq |y|_{(p)}$ is a rearrangement of $\{|y_{i}|\}_{i=1}^p$.
Noting that $\min\limits_{I\subseteq [p],\#I\geq p/2}\|{\bm{y}}_{I}\|_2^2\geq \sum_{i=1}^{\lceil p/2\rceil}|y|_{(i)}^2$,
we have
\begin{align}\label{e:lowerboundp}
\mathbb{P}\Big[\min\limits_{I\subseteq [p],\#I\geq p/2}\frac{1}{p}\|{\bm{y}}_{I}\|_2^2\geq c_-\Big]
\geq \mathbb{P}\Big[\frac{1}{p}\sum_{i=1}^{\lceil p/2\rceil}|y|_{(i)}^2\geq c_-\Big]\overset{(a)}\geq
\mathbb{P}\Big[\frac{1}{p}\sum_{i=1}^{\lceil p/2\rceil}|y|_{(i)}^2\geq (\mu_p-\frac{\nu_0}{2})^2\Big]\overset{(b)}\geq 1-2\exp(-\nu_0^2p/8)
\end{align}
for $c_{-}=\frac{\nu_0^2}{4}$, where the inequalities $(a)$ and $(b)$  follow from
$\mu_p\geq\nu_0$ and Lemma \ref{lem:xu2} with $t=\frac{\nu_0}{2}$, respectively.

Upper bound for \eqref{e:SRIPFIXX}:  Noting that
$$\max\limits_{I\subseteq [p],\#I\geq p/2}\|\mathcal{A}_I(\bm{X})\|_2^2 \leq \|\mathcal{A}(\bm{X})\|_2^2,$$
we apply the inequality \eqref{e:ceninequality} to obtain
\begin{align}\label{e:upperboundp}
\mathbb{P}\Bigg[\frac{1}{p}\|\mathcal{A}(\bm{X})\|_2^2 \leq c_+\Bigg]
\geq 1 - 2\exp\Bigg(-p\left(\frac{t^2}{4} - \frac{t^3}{6}\right)\Bigg),
\end{align}
where \(c_+ = 1 + t\) with \(t \in (0,1)\).

Consequently,  combining \eqref{e:lowerboundp} and \eqref{e:upperboundp},
 we conclude that for any fixed matrix \(\bm{X} \in \mathbb{R}^{n \times m}\) with \(\|\bm{X}\|_F = 1\),  the inequality
\begin{align}\label{e:boundfixed}
c_{-} \leq \min\limits_{I\subseteq [p],\#I\geq p/2}\frac{1}{p}\|\mathcal{A}_{I} (\bm{X})\|_2^2 \leq \max\limits_{I\subseteq [p],\#I\geq p/2}\frac{1}{p}\|\mathcal{A}_{I} (\bm{X})\|_2^2 \leq c_{+}
\end{align}
holds with probability at least \(1 - 4\exp(-cp)\), where  \(c > 0\) is an absolute constant.
For any fixed \(\bm{X} \in \mathbb{R}^{n \times m}\), substituting $\frac{\bm{X}}{\|\bm{X}\|_F}$ into \eqref{e:boundfixed} yields the desired result \eqref{e:SRIPFIXX}.

Step 2: Extension to all rank-$r$ matrices.

For the upper bound of SRIP, since \(\frac{1}{\sqrt{p}}\mathcal{A}\) satisfies the RIP with probability at least \(1 - \exp(-cp)\) provided \(p \geq C(m+n)r\), there exists \(\theta_{+} \in [1,2)\) such that
$$\max\limits_{I\subseteq [p],\#I\geq p/2}\frac{1}{p}\|\mathcal{A}_I(\bm{X})\|_2^2 \leq \frac{1}{p}\|\mathcal{A}(\bm{X})\|_2^2 \leq \theta_{+}\|\bm{X}\|_F^2$$
holds for all matrices \(\bm{X} \in \mathbb{R}^{n\times m}\) of rank at most \(r\) with probability \(1 - \exp(-cp)\) when \(p \geq C(m+n)r\).

We now establish the lower SRIP bound:
there exists \(\theta_{-} \in (0, 2)\) such that
$$\min\limits_{I\subseteq [p],\#I\geq p/2}\frac{1}{p}\|\mathcal{A}_I(\bm{X})\|_2^2 \geq  \theta_{-}\|\bm{X}\|_F^2$$
holds for all matrices \(\bm{X} \in \mathbb{R}^{n\times m}\) of rank at most \(r\) with probability \(1 - \exp(-cp)\) when \(p \geq C(m+n)r\).
 Equivalently, by homogeneity, it suffices to prove that there exists
\(\theta_{-} \in (0,2)\) such that
$$\min\limits_{I\subseteq [p],\#I\geq p/2}\frac{1}{p}\|\mathcal{A}_I(\bm{X})\|_2^2 \geq  \theta_{-}$$
holds for all matrices \(\bm{X} \in U_{r}:=\{\bm{H}\in \mathbb{R}^{n\times m}:\|\bm{H}\|_F=1, \mathrm{rank}(\bm{H})\leq r\}\),  with probability \(1 - \exp(-cp)\) when \(p \geq C(m+n)r\).
Let \(\overline{U_r} \subset U_r\) be an \(\epsilon\)-net of $U_r$ guaranteed by Lemma \ref{cover}, with cardinality \(\#\overline{U_r}= \left(\frac{9}{\epsilon}\right)^{(m+n+1)r}\).
Applying the union bound and \eqref{e:SRIPFIXX}, there exists \(c_- \in (0,2)\) such that
$$c_{-} \leq \min\limits_{I\subseteq [p],\#I\geq p/2}\frac{1}{p}\|\mathcal{A}_{I} (\bm{Y})\|_2^2$$
holds for all $\bm{Y}\in \overline{U_r}$ with probability at least
$1-4(9/\epsilon)^{(m+n+1)r}\exp(-cp)$.
For any \(\bm{X} \in U_r\), there exists some \(\bm{Y} \in \overline{U_r}\) such that \(\|\bm{X} - \bm{Y}\|_F \leq \epsilon\).
 %By the triangle inequality and the upper bound in \eqref{e:SRIPFIXX}, we have
 Note that
\begin{align*}
\min\limits_{I\subseteq [p],\#I\geq p/2}\frac{1}{p}\|\mathcal{A}_{I}(\bm{X})\|_2 &
= \min\limits_{I\subseteq [p],\#I\geq p/2}\frac{1}{p}\|\mathcal{A}_{I}(\bm{X} - \bm{Y}) + \mathcal{A}_{I}(\bm{Y}) \|_2 \\
&\geq \min\limits_{I\subseteq [p],\#I\geq p/2} \frac{1}{p}\|\mathcal{A}_{I}(\bm{Y})\|_2 - \frac{1}{p}\|\mathcal{A}_{I}(\bm{X} - \bm{Y})\|_2 \\
&\geq \sqrt{c_-} - \sqrt{\theta_+}\epsilon
\end{align*}
for \(1 \leq \theta_{+} < 2\). Choosing $\epsilon$ sufficiently small ensures
 \(\sqrt{c_-} - \sqrt{\theta_+}\epsilon > 0\). Thus, there exists \(\theta_- = \sqrt{c_-} - \sqrt{\theta_+}\epsilon \in (0,2)\) such that
$$\min\limits_{I\subseteq [m],\#I\geq p/2}\frac{1}{p}\|\mathcal{A}_{I} (\bm{X})\|_2\geq \theta_-$$
holds for all $\bm{X}\in U_r$ with probability at least
 $$1-4(9/\epsilon)^{(m+n+1)r}\exp(-cp)=1-\exp(-cp+\log 4+(m+n+1)r\log (9/\epsilon)).$$
The desired SRIP property holds for \(\bm{X} \in U_r\) with probability at least \(1 - \exp(-cp)\) provided
$$
p \geq \frac{1}{c}\left(\log 4 + (m+n+1)r \log\left(\frac{9}{\epsilon}\right)\right).
$$

Therefore, \(\mathcal{A}\) satisfies the SRIP property with probability \(1 - \exp(-cp)\) for some absolute constant \(c > 0\), provided \(p \geq C(m+n)r\).
\end{proof}

Lemma \ref{spaceSRIP} assumes that $\mathrm{rank}(\bm{X})\le r$. The same conclusion extends to the broader class $N_r^{*}$ defined in (\ref{setnr}). In particular, by combining the proof strategy of Lemma \ref{spaceSRIP} with the covering-number bound for $N_r^{*}$ established in Lemma \ref{cover2}, we obtain the following result.

\begin{lemma}\label{rmk1}
Let $\mathcal{A}:{\mathbb R}^{n\times m}\to{\mathbb R}^p$ be a Gaussian measurement ensemble and $r \in \mathbb{Z}_{+}$ with $1 \leq r \leq \min\{m,n\}$. For all matrices $\bm{X} \in N_r^{*}:=\{\bm{H}\in \mathbb{R}^{n\times m}:\|\bm{H}\|_F\leq 1, \|\bm{H}\|_{*}\leq \sqrt{r}\}$, there exist constants $0 < \theta_{-} \leq \theta_{+} < 2$ such that the following inequality holds with probability at least $1 - \exp(-cp)$:
$$\theta_{-}\|\bm{X}\|_F^2 \leq \min\limits_{I\subseteq [p],\#I\geq p/2}\frac{1}{p}\| \mathcal{A}_{I}(\bm{X})\|_2^2 \leq \max\limits_{I\subseteq [p],\#I\geq p/2}\frac{1}{p}\|\mathcal{A}_{I}(\bm{X})\|_2^2 \leq \theta_{+}\|\bm{X}\|_F^2,$$
provided that $p \geq C(m+n)r$, where $C$ and $c$ are positive constants. \end{lemma}

%\begin{remark}
%This extension exhibits a  parallel to SRIP results elucidated for the vector set $\{\bm{x} \in \mathbb{R}^n : \|\bm{x}\|_2 \leq 1, \|\bm{x}\|_1 \leq k\}$
%in \cite[Remark II.4]{huang2020estimation}.
%\end{remark}

\section{Proofs of the Main Results}

To facilitate subsequent proofs, we first introduce the following lemmas. Set
$$\sgn(x):=\left\{
                     \begin{array}{ll}
                       1, & \hbox{$x\geq 0$;} \\
                       -1, & \hbox{$x<0$.}
                     \end{array}
                   \right.
$$
\begin{lemma} \label{lemma3.1}
For $\bm{X}_0\in \mathbb{R}^{n\times m}$
 and $\bm{X}\in \mathbb{R}^{n\times m}$,
we define the following sets:
\begin{align*}%\label{set}
T_1:=&\{i\in[p]: \sgn(\langle \bm{A}_i, \bm{X}\rangle)=1, \sgn(\langle \bm{A}_i, \bm{X}_0\rangle)=1\},\nonumber\\
T_2:=&\{i\in[p]: \sgn(\langle \bm{A}_i, \bm{X}\rangle)=-1, \sgn(\langle \bm{A}_i, \bm{X}_0\rangle)=-1\},\nonumber\\
T_3:=&\{i\in[p]: \sgn(\langle \bm{A}_i, \bm{X}\rangle)=1, \sgn(\langle \bm{A}_i, \bm{X}_0\rangle)=-1\},\nonumber\\
T_4:=&\{i\in[p]: \sgn(\langle \bm{A}_i, \bm{X}\rangle)=-1, \sgn(\langle \bm{A}_i, \bm{X}_0\rangle)=1\}.
\end{align*}
Set   $\bm{y}:=\abs{\mathcal{A}(\bm{X}_0)}+\bm{\eta}$.
Then the following inequalities  hold:
\begin{align}\label{imdent1}
\||\mathcal{A}({\bm{X}})|-\bm{y}\|_2^2
\geq
\sum_{i\in T_1\cup T_2} (\langle \bm{A}_i, \bm{H}^{-}\rangle^2+\eta_i^2)
-\sum_{i\in T_1} 2\eta_i\langle \bm{A}_i, \bm{H}^{-}\rangle+\sum_{i\in T_2}2\eta_i\langle \bm{A}_i, \bm{H}^{-}\rangle
\end{align}
and
\begin{align}\label{imdent1+}
\||\mathcal{A}(\bm{X})|-\bm{y}\|_2^2
\geq \sum_{i\in T_3\cup T_4} (\langle \bm{A}_i, \bm{H}^{+}\rangle^2+\eta_i^2)
-\sum_{i\in T_3} 2\eta_i\langle \bm{A}_i, \bm{H}^{+}\rangle+\sum_{i\in T_4}2\eta_i\langle \bm{A}_i, \bm{H}^+\rangle,
\end{align}
where $\bm{H}^{-}:=\bm{X}-\bm{X}_0$ and $\bm{H}^{+}:=\bm{X}+\bm{X}_0$.
\end{lemma}
\begin{proof}
The  proof is presented in Appendix \ref{prof:lemma3.1}.
\end{proof}
\begin{remark}\label{remark3.2}
By definition, the sets \(T_i\) \((i=1,2,3,4)\) are pairwise disjoint and satisfy \(T_1 \cup T_2 \cup T_3 \cup T_4 = [p]\). Consequently, either \(|T_1 \cup T_2| \ge p/2\) or \(|T_3 \cup T_4| \ge p/2\).
\end{remark}

\begin{lemma}
\label{lem:upper1}
Let $\mathcal{A}:{\mathbb R}^{n\times m}\to{\mathbb R}^p$ be a Gaussian measurement ensemble and $r \in \mathbb{Z}_{+}$ with $1 \leq r \leq \min\{m,n\}$.
Suppose \( p \gtrsim (\sqrt{m} + \sqrt{n})^2 r \) for any fixed $r$.  Then, with probability at least \( 1 - \exp(-cp) \), the following inequality holds for any fixed \( \bm{\eta} \in \mathbb{R}^p \):
\[
\sup_{\bm{H} \in N_r^{*}, T \subset [p]} \sum_{i \in T} \eta_i \langle \bm{A}_i, \bm{H} \rangle\,\, \lesssim\,\, \sqrt{p} \|\bm{\eta}\|_2,
\]
where $N_r^{*}:=\{\bm{H}\in \mathbb{R}^{n\times m}:\|\bm{H}\|_F\leq 1, \|\bm{H}\|_{*}\leq \sqrt{r}\}$, and  \( c > 0 \) is an absolute constant.
\end{lemma}
\begin{proof}
The proof is provided in Appendix \ref{prof:upper1}.
\end{proof}

\subsection{Proofs for Rank-constrained Least Squares Model}
Before turning to the proofs, we establish the following lemma for later use.
\begin{lemma}\label{lem:unbound}
Let $\mathcal{A}:\mathbb{R}^{n\times m}\to \mathbb{R}^p$ be a Gaussian measurement ensemble. Let $r$ be an integer with $1\le r\le \min\{m,n\}$,  and suppose $p \gtrsim (\sqrt{m}+\sqrt{n})^{2} r$. Let $\bm{\eta}=(\eta_1,\ldots,\eta_p)^{\top}\in \mathbb{R}^p$ denote the noise vector.
Then we have

$(\mathrm{I})$ The inequality
\begin{align*}
\Big|\sum_{i=1}^p\eta_i|\langle  \bm{A}_{i},\bm{X}\rangle|\Big|
\lesssim \sqrt{(m+n+1)r\log p}\|\bm{\eta}\|_2+\Big|\sum_{i=1}^p\eta_i\Big|
\end{align*}
holds uniformly for all $\bm{X}\in U_r:=\{\bm{H}\in \mathbb{R}^{n\times m}:\|\bm{H}\|_F=1, \mathrm{rank}(\bm{H})\leq r\}$, with probability at least $1 - \exp(-c p) - 2\exp\big(-(m+n+1)r\log p\big)$.
Here $c>0$ is a universal constant.

$(\mathrm{II})$ If $\Big|\sum_{i=1}^p\eta_i\Big|\geq C_0\sqrt{p}\|\bm{\eta}\|_2$ with $C_0\in (0,1)$, then
$$\Big|\sum_{i=1}^p\eta_i|\langle  \bm{A}_{i},\bm{X}\rangle|\Big|
\gtrsim\sqrt{p}\|\bm{\eta}\|_2
$$
holds for all $\bm{X}\in U_r:=\{\bm{H}\in \mathbb{R}^{n\times m}:\|\bm{H}\|_F=1, \mathrm{rank}(\bm{H})\leq r\}$, with probability at least $1-3\exp(-cp)$. Here $c>0$ is a universal constant.
\end{lemma}
\begin{proof}
The proof is provided in Appendix \ref{prof:unbound}.
\end{proof}
\subsection*{A. Proof of Theorem \ref{thm:squarank}}
\begin{proof}
Given that $\hat{\bm{X}}^r$ is a solution to \eqref{pro:squarank} and $\bm{y}=|\mathcal{A}(\bm{X}_0)|+\bm{\eta}$, it follows that
\begin{align}\label{2}
\||\mathcal{A}(\hat{\bm{X}}^r)|-\bm{y}\|_2^2\leq \||\mathcal{A}(\bm{X}_0)|-\bm{y}\|_2^2=\|\bm{\eta}\|_2^2
\end{align}
and $\mathrm{rank}(\hat{\bm{X}}^r)\leq r$.
When applying Lemma \ref{lemma3.1}, we set $\bm{X}:=\hat{\bm{X}}^r$. This choice then consistently defines $\bm{H}^{-}=\hat{\bm{X}}^r-\bm{X}_0$ and $\bm{H}^{+}=\hat{\bm{X}}^r+\bm{X}_0$ according to the lemma's specifications.
 Given that  $\mathrm{rank}(\bm{X}_0)\leq r$ and $\mathrm{rank}(\hat{\bm{X}}^r)\leq r$,
  we have
  \[
  \frac{\bm{H}^{-}}{\|\bm{H}^{-}\|_F} \in U_{2r}:=\{\bm{H}\in \mathbb{R}^{n\times m}:\|\bm{H}\|_F=1, \mathrm{rank}(\bm{H})\leq 2r\}.
  \]
   %, where $U_{2r}$ is defined in \eqref{setur}.
% \textcolor{red}{Following this step, we will continue to employ the symbol $\hat{\bm{X}}^r$ for this quantity.}
In accordance with  Remark \ref{remark3.2},  we assume, without loss of generality, that $\sharp(T_1\cup T_2)\geq \frac p 2$.
Combining  \eqref{2} and \eqref{imdent1}, we obtain that
\begin{align}\label{eq:le1}
\sum_{i\in T_1\cup T_2} \langle \bm{A}_i, \bm{H}^{-}\rangle^2\leq
\sum_{i\in T_1} 2\eta_i\langle \bm{A}_i, \bm{H}^{-}\rangle-\sum_{i\in T_2}2\eta_i\langle \bm{A}_i, \bm{H}^{-}\rangle+\sum_{i\in (T_1\cup T_2)^c}\eta_i^2.
\end{align}
Applying Lemma \ref{spaceSRIP} to the rank $2r$ matrix $\bm{H}^{-}$, we obtain
\begin{align}\label{eq:le15}
 \sum_{i\in T_1\cup T_2} \langle \bm{A}_i, \bm{H}^{-}\rangle^2=\|\mathcal{A}_{T_1\cup T_2}(\bm{H}^{-})\|_2^2\geq p\theta_-\|\bm{H}^{-}\|_F^2
 \end{align}
 holds with probability at least $1-\exp(-cp)$, provided $p\gtrsim (m+n)r$.
Note that
\begin{align}\label{eq:le2}
 \sum_{i\in T_1}\eta_i\langle \bm{A}_i, \bm{H}^{-}\rangle-\sum_{i\in T_2}\eta_i\langle \bm{A}_i, \bm{H}^{-}\rangle
\leq&  \Big|\sum_{i\in T_1}\eta_i\langle \bm{A}_i, \bm{H}^{-}\rangle\Big|+\Big|\sum_{i\in T_2}\eta_i\langle \bm{A}_i, \bm{H}^{-}\rangle\Big|\nonumber \\
\overset{(a)}{\leq}& 2\|\bm{H}^{-}\|_F\sup_{\bm{H} \in N_{2r}^{*}, T \subset [p]} \sum_{i \in T} \eta_i \langle \bm{A}_i, \bm{H} \rangle \nonumber \\
 \overset{(b)}\leq& C\sqrt{p} \|\bm{\eta}\|_2\|\bm{H}^{-}\|_F
 \end{align}
holds with probability at least $1-\exp(-cp)$, provided $ p \gtrsim (\sqrt{m} + \sqrt{n})^2 r$.
 Here, ($a$) follows from the fact that
 \[
 \frac{\bm{H}^{-}}{\|\bm{H}^{-}\|_F} \in U_{2r}\subset N_{2r}^{*}:=\{\bm{H}\in \mathbb{R}^{n\times m}:\|\bm{H}\|_F\leq 1, \|\bm{H}\|_{*}\leq \sqrt{2r}\},
  \]
  while ($b$) follows from Lemma \ref{lem:upper1}.

 Combining (\ref{eq:le1}), (\ref{eq:le15}) and (\ref{eq:le2}), we conclude that
 $$
 p\theta_-\|\bm{H}^{-}\|_F^2\leq C\sqrt{p}\|\bm{\eta}\|_2\|\bm{H}^{-}\|_F+\sum_{i\in (T_1\cup T_2)^c}\eta_i^2
 $$
holds  with probability at least  $1-2\exp(-cp)$, provided $p\gtrsim (m+n)r$.
 This implies
 \begin{align}\label{e:H-}
 \|\bm{H}^{-}\|_F\lesssim \frac{\|\bm{\eta}\|_2}{\sqrt{p}}.
 \end{align}
 For the case where $\sharp(T_3 \cup T_4) \geq \frac{p}{2}$, employing a similar line of reasoning as in the previous scenario, we can deduce
\begin{align}\label{e:H+}
 \|\bm{H}^{+}\|_F \lesssim \frac{\|\bm{\eta}\|_2}{\sqrt{p}}.
 \end{align}
%where $\bm{H}^{+} = \hat{\bm{X}}^{r} + \bm{X}_0$.
Combining \eqref{e:H-} and \eqref{e:H+}, we have
$$\min\{\|\hat{\bm{X}}^r-\bm{X}_0\|_F, \|\hat{\bm{X}}^r+\bm{X}_0\|_F\}\lesssim \frac{\|\bm{\eta}\|_2}{\sqrt{p}},$$
 which is \eqref{eq:upper}.
\end{proof}

\subsection*{B. Proof of Theorems \ref{thm:squarankone}}

We now present the proof of Theorem \ref{thm:squarankone}, inspired by the proof of \cite[Theorem 1.4]{xia2024performance}.
\begin{proof}
For the case where \(\hat{\bm{X}}^r = \bm{0}\in \mathbb{R}^{n\times m}\), i.e., \(\mathrm{rank}(\hat{\bm{X}}^r) = 0\),
the conclusion follows directly from the condition \(\|\bm{X}_0\|_{F} \gtrsim \frac{\|\bm{\eta}\|_2}{\sqrt{p}}\).

Next, we consider the case where \(1 \leq \mathrm{rank}(\hat{\bm{X}}^r) \leq r\). Without loss of generality, we assume that $\mathrm{rank}(\hat{\bm{X}}^r) = r$. In this case,
the rank-$1$ decomposition of $\hat{\bm{X}}^r\in \mathbb{R}^{n\times m}$ can be expressed as $\hat{\bm{X}}^r=\sum_{i=1}^{r}\hat{\bm{u}}_i\hat{\bm{v}}_i^{\top}$, where $\hat{\bm{u}}_i\in \mathbb{R}^n$ and $\hat{\bm{v}}_i\in \mathbb{R}^m$ are nonzero vectors for all  \(i \in [r]\).
We now claim  that $\{\hat{\bm{u}}_i,\hat{\bm{v}}_i\}_{i\in [r]}$ forms a solution to the following optimization problem
\begin{align}\label{pro:squarankv}
\min_{\bm{u}_i\in \mathbb{R}^n, \bm{v}_i\in \mathbb{R}^m, i\in [r]} \ \ \ \Big\|\Big|\mathcal{A}\Big(\sum_{i=1}^{r}\bm{u}_i\bm{v}_i^{\top}\Big)\Big|-\bm{y}\Big\|_2^2.
\end{align}
The proof of this claim  is provided at the end of this argument.

Consider the function
\begin{align*}
f(\bm{u}_1,\bm{v}_1,\ldots,\bm{u}_r,\bm{v}_r)=\Big\|\Big|\mathcal{A}\Big(\sum_{i=1}^{r}\bm{u}_i\bm{v}_i^{\top}\Big)\Big|-\bm{y}\Big\|_2^2=\sum_{i=1}^{p}\bigg(\sqrt{\Big \langle \bm{A}_i, \sum_{j=1}^{r}\bm{u}_j\bm{v}_j^{\top}\Big\rangle^2}-y_i\bigg)^2
\end{align*}
for all $\bm{u}_i\in \mathbb{R}^n, \bm{v}_i\in \mathbb{R}^m$.
The sub-gradient of this function is given by
\begin{align*}
\frac{\partial f(\bm{u}_1,\bm{v}_1,\cdots,\bm{u}_r,\bm{v}_r)}{\partial (\bm{u}_1,\bm{v}_1,\cdots,\bm{u}_r,\bm{v}_r)}=\left(
  \begin{array}{c}
 \frac{\partial f(\bm{u}_1,\bm{v}_1,\cdots,\bm{u}_r,\bm{v}_r)}{\partial \bm{u}_1 } \\
 \frac{\partial f(\bm{u}_1,\bm{v}_1,\cdots,\bm{u}_r,\bm{v}_r)}{\partial  \bm{v}_1} \\
\vdots\\
\frac{\partial f(\bm{u}_1,\bm{v}_1,\cdots,\bm{u}_r,\bm{v}_r)}{\partial \bm{u}_r } \\
 \frac{\partial f(\bm{u}_1,\bm{v}_1,\cdots,\bm{u}_r,\bm{v}_r)}{\partial  \bm{v}_r} \\
  \end{array}
\right)=\left(
  \begin{array}{c}
 2\sum_{i=1}^{p}(|\langle \bm{A}_i, \sum_{j=1}^{r}\bm{u}_j\bm{v}_j^{\top}\rangle|-y_i)z_i\bm{A}_{i}\bm{v}_1 \\
 2\sum_{i=1}^{p}(|\langle\bm{A}_i, \sum_{j=1}^{r}\bm{u}_j\bm{v}_j^{\top}\rangle|-y_i)z_i\bm{A}_{i}^{\top}\bm{u}_1 \\
\vdots\\
2\sum_{i=1}^{p}(|\langle \bm{A}_i, \sum_{j=1}^{r}\bm{u}_j\bm{v}_j^{\top}\rangle|-y_i)z_i\bm{A}_{i}\bm{v}_r \\
 2\sum_{i=1}^{p}(|\langle\bm{A}_i, \sum_{j=1}^{r}\bm{u}_j\bm{v}_j^{\top}\rangle|-y_i)z_i\bm{A}_{i}^{\top}\bm{u}_r \\
  \end{array}
\right),
\end{align*}
where
\begin{equation}\label{eq:zii}
z_{i}=\left\{
                \begin{array}{ll}
                  \sgn(\langle \bm{A}_i, \sum_{j=1}^{r}\bm{u}_i\bm{v}_i^{\top}\rangle),
 & \hbox{$\mathrm{if}$ $|\langle \bm{A}_i, \sum_{j=1}^{r}\bm{u}_i\bm{v}_i^{\top}\rangle|\neq 0$;} \\
                 \in  [-1, 1], & \hbox{other wise.}
                \end{array}
              \right.
\end{equation}
Since $\{\hat{\bm{u}}_i,\hat{\bm{v}}_i\}_{i\in [r]}$  with $\hat{\bm{X}}^r=\sum_{i=1}^{r}\hat{\bm{u}}_i\hat{\bm{v}}_i^{\top}$ is a solution to \eqref{pro:squarankv}, we have
\[
\bm{0}\in \frac{\partial f(\bm{u}_1,\bm{v}_2,\cdots,\bm{u}_r,\bm{v}_r)}{\partial (\bm{u}_1,\bm{v}_2,\cdots,\bm{u}_r,\bm{v}_r)}|_{\bm{u}_1=\hat{\bm{u}}_1,\cdots,\bm{v}_r=\hat{\bm{v}}_r}.
\]
Hence, there exists $\hat{z}_i$ with $i\in [r]$, satisfying (\ref{eq:zii}) such that
\begin{align*}
\bm{0}=\left(
  \begin{array}{c}
 2\sum_{i=1}^{p}(|\langle \bm{A}_i, \sum_{j=1}^{r}\hat{\bm{u}}_j\hat{\bm{v}}_j^{\top}\rangle|-y_i)\hat{z}_i\bm{A}_{i}\hat{\bm{v}}_1 \\
 2\sum_{i=1}^{p}(|\langle\bm{A}_i, \sum_{j=1}^{r}\hat{\bm{u}}_j\hat{\bm{v}}_j^{\top}\rangle|-y_i)\hat{z}_i\bm{A}_{i}^{\top}\hat{\bm{u}}_1 \\
\vdots\\
2\sum_{i=1}^{p}(|\langle \bm{A}_i, \sum_{j=1}^{r}\hat{\bm{u}}_j\hat{\bm{v}}_j^{\top}\rangle|-y_i)\hat{z}_i\bm{A}_{i}\hat{\bm{v}}_r \\
 2\sum_{i=1}^{p}(|\langle\bm{A}_i, \sum_{j=1}^{r}\hat{\bm{u}}_j\hat{\bm{v}}_j^{\top}\rangle|-y_i)\hat{z}_i\bm{A}_{i}^{\top}\hat{\bm{u}}_r \\
  \end{array}
\right).
\end{align*}
By direct calculation, we have
\begin{align*}
0=&\Big\langle
 \Big(\sum_{i=1}^{p}(|\langle \bm{A}_i, \sum_{j=1}^{r}\hat{\bm{u}}_j\hat{\bm{v}}_j^{\top}\rangle|-y_i)\hat{z}_i\bm{A}_{i}\hat{\bm{v}}_l\Big)
\Big(\sum_{i=1}^{p}(|\langle\bm{A}_i,\sum_{j=1}^{r}\hat{\bm{u}}_j\hat{\bm{v}}_j^{\top}\rangle|-y_i)\hat{z}_i\bm{A}_{i}^{\top}\hat{\bm{u}}_l\Big)^{\top}, \hat{\bm{u}}_l\hat{\bm{v}}_l^{\top}\Big\rangle\\
%=&\mathrm{Tr}\Big(\Big(\sum_{i=1}^{p}(|\langle\bm{A}_i,\sum_{j=1}^{r}\hat{\bm{u}}_j\hat{\bm{v}}_j^{\top}\rangle|-y_i)\hat{z}_i\bm{A}_{i}^{\top}\hat{\bm{u}}_l\Big)\Big(\sum_{i=1}^{p}(|\langle \bm{A}_i, \sum_{j=1}^{r}\hat{\bm{u}}_j\hat{\bm{v}}_j^{\top}\rangle|-y_i)\hat{z}_i\bm{A}_{i}\hat{\bm{v}}_l\Big)^{\top}\hat{\bm{u}}_l\hat{\bm{v}}_l^{\top}\Big)\\
=&\mathrm{Tr}\Big(\hat{\bm{v}}_l^{\top}\Big(\sum_{i=1}^{p}(|\langle\bm{A}_i,\sum_{j=1}^{r}\hat{\bm{u}}_j\hat{\bm{v}}_j^{\top}\rangle|-y_i)\hat{z}_i
\bm{A}_{i}^{\top}\hat{\bm{u}}_l\Big)\Big(\sum_{i=1}^{p}(|\langle \bm{A}_i, \sum_{j=1}^{r}\hat{\bm{u}}_j\hat{\bm{v}}_j^{\top}\rangle|-y_i)\hat{z}_i\bm{A}_{i}\hat{\bm{v}}_l\Big)^{\top}\hat{\bm{u}}_l\Big)\\
=&\Big(\sum_{i=1}^{p}(|\langle\bm{A}_i,\sum_{j=1}^{r}\hat{\bm{u}}_j\hat{\bm{v}}_j^{\top}\rangle|-y_i)\hat{z}_i\langle \bm{A}_i,\hat{\bm{u}}_l\hat{\bm{v}}_l^{\top}\rangle
\Big)^2
\end{align*}
for all $l\in[r]$.
Then, we obtain that
\begin{align*}
0=&\sum_{l=1}^r\Big(\sum_{i=1}^{p}(|\langle\bm{A}_i,\sum_{j=1}^{r}\hat{\bm{u}}_j\hat{\bm{v}}_j^{\top}\rangle|-y_i)\hat{z}_i\langle \bm{A}_i,\hat{\bm{u}}_l\hat{\bm{v}}_l^{\top}\rangle\Big)
=\sum_{i=1}^{p}(|\langle\bm{A}_i,\sum_{j=1}^{r}\hat{\bm{u}}_j
\hat{\bm{v}}_j^{\top}\rangle|-y_i)\cdot |\langle \bm{A}_i,\sum_{l=1}^{r}\hat{\bm{u}}_l\hat{\bm{v}}_l^{\top}\rangle|,
\end{align*}
where the last equality is based on  $\hat{z}_i$ satisfying (\ref{eq:zii}).
That is,
 \begin{align}\label{e:equ0}
\sum_{i=1}^{p}(|\langle \bm{A}_i, \hat{\bm{X}}^r\rangle|-y_i)\cdot |\langle \bm{A}_{i},\hat{\bm{X}}^r\rangle|=0.
\end{align}
Substituting $y_i=|\langle \bm{A}_i, \bm{X}_0\rangle|+\eta_i$ with $i\in [p]$ into the  equality \eqref{e:equ0}, we obtain
\[
\sum_{i=1}^{p}\eta_i|\langle \bm{A}_{i},\hat{\bm{X}}^r\rangle|
=\sum_{i=1}^{p}(|\langle \bm{A}_i, \hat{\bm{X}}^r\rangle|-|\langle \bm{A}_i, \bm{X}_0\rangle|)|\langle \bm{A}_{i},\hat{\bm{X}}^r\rangle|.
\]
Hence we obtain that
\begin{align*}
\sqrt{p}\|\bm{\eta}\|_2\|\hat{\bm{X}}^r\|_F\overset{(a)}\lesssim \left|\sum_{i=1}^{p}\eta_i|\langle \bm{A}_{i},\hat{\bm{X}}^r\rangle|\right|
=&\sum_{i=1}^{p}\left ||\langle \bm{A}_i, \hat{\bm{X}}^r\rangle|-|\langle \bm{A}_i, \bm{X}_0\rangle|\right||\langle \bm{A}_{i},\hat{\bm{X}}^r\rangle|\\
{\leq}&\sum_{i=1}^{p}\min\{|\langle \bm{A}_i, \hat{\bm{X}}^r-\bm{X}_0\rangle|,|\langle \bm{A}_i, \hat{\bm{X}}^r+\bm{X}_0\rangle|\}|\langle \bm{A}_{i},\hat{\bm{X}}^r\rangle|\\
\leq&\min\{\|\mathcal{A}(\hat{\bm{X}}^r-\bm{X}_0)\|_2, \|\mathcal{A}(\hat{\bm{X}}^r+\bm{X}_0)\|_2\} \|\mathcal{A}(\hat{\bm{X}}^r)\|_2\\
\overset{(b)}\lesssim & p \min \{\|\hat{\bm{X}}^r-\bm{X}_0\|_F, \|\hat{\bm{X}}^r+\bm{X}_0\|_F\}\|\hat{\bm{X}}^r\|_F,
\end{align*}
holds with probability at least $1-5\exp(-cp)$, where $(a)$ is derived from Lemma \ref{lem:unbound} (II), and $(b)$ is derived from Lemma \ref{spaceSRIP}, respectively.
 Therefore, we have
$$\min\{\|\hat{\bm{X}}^r-\bm{X}_0\|_F, \|\hat{\bm{X}}^r+\bm{X}_0\|_F\} \gtrsim \frac{\|\bm{\eta}\|_2}{\sqrt{p}},$$
which leads to the conclusion.

It remains to prove the claim about the solution to  \eqref{pro:squarankv}.
Suppose there exists a set $\{\tilde{\bm{v}}_i,\tilde{\bm{u}}_i\}_{i\in[r]}$ such that
$$\Big\|\Big|\mathcal{A}\Big(\sum_{i=1}^{r}\tilde{\bm{u}}_i\tilde{\bm{v}}_i^{\top}\Big)\Big|-\bm{y}\Big\|_2^2
<\Big\|\Big|\mathcal{A}\Big(\sum_{i=1}^{r}\hat{\bm{u}}_i\hat{\bm{v}}_i^{\top}\Big)\Big|-\bm{y}\Big\|_2^2.$$
Then, it follows that
$$\||\mathcal{A}(\tilde{\bm{X}})|-\bm{y}\|_2^2<\||\mathcal{A}(\hat{\bm{X}}^r)|-\bm{y}\|_2^2,$$
where $\tilde{\bm{X}}=\sum_{i=1}^r\tilde{\bm{u}}_i\tilde{\bm{v}}_i^{\top}$.
This leads to a contradiction with the assumption that $\hat{\bm{X}}^r$ is the solution to \eqref{pro:squarank}.
\end{proof}

\subsection*{C. Proof of Theorem \ref{thm:squaranks}}
\begin{proof}
Given that $\hat{\bm{X}}^r$ is a solution to \eqref{pro:squarank}, we have
\begin{align}\label{e:rankmin}
\||\mathcal{A}(\hat{\bm{X}}^r)|-\bm{y}\|_2^2\leq \||\mathcal{A}(\bm{X}_0)|-\bm{y}\|_2^2.
\end{align}
Substituting  $\bm{y}=|\mathcal{A}(\bm{X}_0)|+\bm{\eta}$ into \eqref{e:rankmin},  we have
\begin{align}\label{e:ineq1}
\sum_{i=1}^p(|\langle \bm{A}_i, \hat{\bm{X}}^r\rangle|-|\langle \bm{A}_i, \bm{X}_0\rangle|)^2
\leq\sum_{i=1}^p 2\eta_i(|\langle \bm{A}_i, \hat{\bm{X}}^r\rangle|-|\langle \bm{A}_i, \bm{X}_0\rangle|).
\end{align}
Let's consider the left-hand side of \eqref{e:ineq1}. With probability at least $1-e^{-cp}$, we have
\begin{align}\label{e:left}
\sum_{i=1}^p(|\langle \bm{A}_i, \hat{\bm{X}}^r\rangle|-|\langle \bm{A}_i, \bm{X}_0\rangle|)^2
\overset{(a)}\geq \sum_{i\in T_1\cup T_2}\left|\langle \bm{A}_i, \hat{\bm{X}}^r-\bm{X}_0\rangle\right|^2 \overset{(b)}\geq p\theta_{-}\|\hat{\bm{X}}^r-\bm{X}_0\|_F^2.
\end{align}
The sets $T_1$ and $T_2$ are defined as in Lemma \ref{lemma3.1}, by setting $\bm{X}:=\hat{\bm{X}}^r$ as done in the proof of Theorem \ref{thm:squarank}.
  By Remark \ref{remark3.2}, we can, without loss of generality, again assume $\sharp(T_1\cup T_2)\geq \frac p 2$.
  Furthermore, since $\mathrm{rank}(\bm{X}_0)\leq r$ and $\mathrm{rank}(\hat{\bm{X}}^r)\leq r$, we have
  \[
  \frac{\hat{\bm{X}}^r-\bm{X}_0}{\|\hat{\bm{X}}^r-\bm{X}_0\|_F} \in U_{2r}:=\{\bm{H}\in \mathbb{R}^{n\times m}:\|\bm{H}\|_F=1, \mathrm{rank}(\bm{H})\leq 2r\} .
  \]
  %, where $U_{2r}$ is defined in \eqref{setur}.
 The  inequalities $(a)$ and $(b)$
 follow from $T_1\cup T_2\subset[p]$ and Lemma \ref{spaceSRIP}, respectively.

We next turn to  the right-hand side of the inequality \eqref{e:ineq1}.  With probability at least $1- 2e^{-cp}-4e^{-(m+n+1)r\log p}$,  we have
\begin{align}\label{eq:shang1}
\sum_{i=1}^p 2\eta_i(|\langle \bm{A}_i, \hat{\bm{X}}^r\rangle|-|\langle \bm{A}_i, \bm{X}_0\rangle|)
\leq &\Big| \sum_{i=1}^p 2\eta_i|\langle \bm{A}_i, \hat{\bm{X}}^r\rangle|\Big|+ \Big|\sum_{i=1}^p 2\eta_i|\langle \bm{A}_i, \bm{X}_0\rangle|\Big|\nonumber\\
\overset{(a)}\lesssim  &\Big(\sqrt{(m+n+1)r\log p}\|\bm{\eta}\|_2+\Big|\sum_{i=1}^p\eta_i\Big|\Big)(\|\hat{\bm{X}}^r\|_F+\|\bm{X}_0\|_F)
\end{align}
where the inequality $(a)$ is based on Lemma \ref{lem:unbound} (I).
Substituting \eqref{eq:shang1} and \eqref{e:left} into \eqref{e:ineq1}, we obtain
\begin{align}\label{e:ineq3}
p\theta_{-}\|\hat{\bm{X}}^r-\bm{X}_0\|_F^2 \leq C\Big(\sqrt{(m+n+1)r\log p}\|\bm{\eta}\|_2+\Big|\sum_{i=1}^p\eta_i\Big|\Big)(\|\hat{\bm{X}}^r\|_F+\|\bm{X}_0\|_F),
\end{align}
where $C>0$ is some absolute constant.
Thus, we have
\begin{align}\label{e:ineq2}
(\|\hat{\bm{X}}^r\|_F-\|\bm{X}_0\|_F)^2 \leq \frac{C}{p\theta_{-}} \Big(\sqrt{(m+n+1)r\log p}\|\bm{\eta}\|_2+\Big|\sum_{i=1}^p\eta_i\Big|\Big)(\|\hat{\bm{X}}^r\|_F+\|\bm{X}_0\|_F).
\end{align}
Note that the inequality \( a \leq 3b + l \) holds under the condition \( (a - b)^2 \leq l(a + b) \) for \( a, b, l \geq 0 \), as found in \cite[(4.9)]{xia2024performance}. Applying this to  \eqref{e:ineq2} with
$a:=\|\hat{\bm{X}}^r\|_F$ and $b:=\|\bm{X}_0\|_F$, we have
\begin{equation}\label{eq:budeng0}
\|\hat{\bm{X}}^r\|_F\leq \frac{C}{p\theta_{-}} \Big(\sqrt{(m+n+1)r\log p}\|\bm{\eta}\|_2+\Big|\sum_{i=1}^p\eta_i\Big|\Big)+3\|\bm{X}_0\|_F.
\end{equation}
Substituting (\ref{eq:budeng0}) into \eqref{e:ineq3},  we obtain that
\begin{align*}
\|\hat{\bm{X}}^r-\bm{X}_0\|_F^2 \leq
\frac{C}{p\theta_{-}}\Big(\sqrt{(m+n+1)r\log p}\|\bm{\eta}\|_2+\Big|\sum_{i=1}^p\eta_i\Big|\Big)\Big( \frac{C}{p\theta_{-}} \Big(\sqrt{(m+n+1)r\log p}\|\bm{\eta}\|_2+\Big|\sum_{i=1}^p\eta_i\Big|\Big)+4\|\bm{X}_0\|_F\Big).
\end{align*}
Therefore, by substituting condition \eqref{e:noise} into the preceding inequality, we conclude that
\begin{align*}
\|\hat{\bm{X}}^r-\bm{X}_0\|_F \lesssim \max\Big\{ \|\bm{X}_0\|_F,\sqrt{\frac{(m+n+1)r\log p}{p}}\Big\}\sqrt{\frac{(m+n+1)r\log p}{p}}.
\end{align*}
This holds with probability at least $1- 3\exp(-cp)-4\exp(-(m+n+1)r\log p)$, where $c>0$ is an absolute constant.
For the case where $\sharp(T_1\cup T_2)\geq \frac p 2$, following a similar derivation as above, we obtain that
\begin{align*}
\|\hat{\bm{X}}^r+\bm{X}_0\|_F \lesssim \max\Big\{ \|\bm{X}_0\|_F,\sqrt{\frac{(m+n+1)r\log p}{p}}\Big\}\sqrt{\frac{(m+n+1)r\log p}{p}}
\end{align*}
holds with probability at least $1- 3\exp(-cp)-4\exp(-(m+n+1)r\log p)$.
Combining the results from both cases immediately yields the desired conclusion.
\end{proof}

\subsection{Proofs for Nuclear Norm-based Models}
The proofs of several theorems for nuclear norm-based models
rely heavily on key properties of the nuclear norm and specific matrix decomposition techniques. These fundamental concepts are encapsulated in the following lemmas, which are useful in our subsequent analyses.
\begin{lemma} \label{lemma 3.6}\cite[Lemma 3.4]{recht2010guaranteed}
 Let $\bm{B}, \bm{C} \in \mathbb{R}^{n\times m}$. Then there exist
matrices $\bm{C}_1, \bm{C}_2 \in \mathbb{R}^{n\times m}$  such that
$$ \bm{C} = \bm{C}_1 + \bm{C}_2,\ \ \mathrm{rank}(\bm{C}_1)\leq 2\cdot \mathrm{rank}(\bm{B}),\ \
\bm{B}\bm{C}_2^{\top}=\bm{0}, \ \ \bm{B}^{\top}\bm{C}_2=\bm{0},\ \
 \langle \bm{C}_1,\bm{C}_2\rangle=0.$$
\end{lemma}

\begin{lemma}\label{lemma3.7}\cite[Lemma 2.3]{recht2010guaranteed}
Let $\bm{B}, \bm{C} \in \mathbb{R}^{n\times m}$.
If $\bm{B}\bm{C}^{\top}=\bm{0}$ and $\bm{C}^{\top}\bm{B}=\bm{0}$,
then $\|\bm{B} +\bm{C}\|_{*} = \|\bm{B}\|_{*}+ \|\bm{C}\|_{*}$.
\end{lemma}

\subsection*{A. Proof of Theorem \ref{thm:rankmin}}
\begin{proof}
Given  $\bm{y}=|\mathcal{A}(\bm{X}_0)|+\bm{\eta}$ and  $\hat{\bm{X}}^{\ell_2}$ is a solution to \eqref{leastproblemnu} with $R=\|\bm{X}_0\|_{*}$, we have
\begin{align}\label{2nu}
\||\mathcal{A}(\hat{\bm{X}}^{\ell_2})|-\bm{y}\|_2^2\leq \||\mathcal{A}(\bm{X}_0)|-\bm{y}\|_2^2=\|\bm{\eta}\|_2^2\ \ \ \   \mathrm{and}\ \ \ \
\|\hat{\bm{X}}^{\ell_2}\|_{*}\leq \|\bm{X}_0\|_{*}.
\end{align}

Next, we apply Lemma \ref{lemma3.1} by setting $\bm{X}:=\hat{\bm{X}}^{\ell_2}$. This implies $\bm{H}^{-}=\hat{\bm{X}}^{\ell_2}-\bm{X}_0$ and $\bm{H}^{+}=\hat{\bm{X}}^{\ell_2}+\bm{X}_0$, with the sets $T_1, T_2, T_3, T_4$ defined as specified in this lemma.
As noted in Remark \ref{remark3.2}, we assume, without loss of generality, that $\sharp(T_1 \cup T_2) \geq \frac{p}{2}$. We claim that the following holds:
  \begin{align}\label{e:setH}
  \frac{\bm{H}^{-}}{2\sqrt{2}\|\bm{H}^{-}\|_F} \in N_{r}^{*}:
=\{\bm{H}\in \mathbb{R}^{n\times m}:\|\bm{H}\|_F\leq 1, \|\bm{H}\|_{*}\leq \sqrt{r}\}.
\end{align}
This claim will be established at the end of the proof.

Substituting \eqref{2nu} into \eqref{imdent1} in Lemma \ref{lemma3.1}, we obtain
\begin{align}\label{e:AT12}
\sum_{i\in T_1\cup T_2} \langle \bm{A}_i, \bm{H}^{-}\rangle^2\leq
\sum_{i\in T_1} 2\eta_i\langle \bm{A}_i, \bm{H}^{-}\rangle-\sum_{i\in T_2}2\eta_i\langle \bm{A}_i, \bm{H}^{-}\rangle+\sum_{i\in (T_1\cup T_2)^c}\eta_i^2.
\end{align}
For the left-hand side of \eqref{e:AT12}, by applying Lemma \ref{rmk1} and \eqref{e:setH}, the inequality
\begin{align}\label{e:add1}
 \sum_{i\in T_1\cup T_2} \langle \bm{A}_i, \bm{H}^{-}\rangle^2=\|\mathcal{A}_{T_1\cup T_2}(\bm{H}^{-})\|_2^2\geq p\theta_-\|\bm{H}^{-}\|_F^2
 \end{align}
 is satisfied with probability at least $1-\exp(-cp)$, provided $p\gtrsim (m+n)r$.
Next, for the right-hand side of \eqref{e:AT12}, utilizing Lemma \ref{lem:upper1} and \eqref{e:setH}, we derive that
 \begin{align}\label{e:add2}
 \sum_{i\in T_1}\eta_i\langle \bm{A}_i, \bm{H}^{-}\rangle-\sum_{i\in T_2}\eta_i\langle \bm{A}_i, \bm{H}^{-}\rangle
\leq&  \Big|\sum_{i\in T_1}\eta_i\langle \bm{A}_i, \bm{H}^{-}\rangle\Big|+\Big|\sum_{i\in T_2}\eta_i\langle \bm{A}_i, \bm{H}^{-}\rangle\Big|\nonumber\\
\leq& 2\sqrt{2}\|\bm{H}^{-}\|_F\sup_{\bm{H} \in N_r^{*}, T \subset [p]} \sum_{i \in T} \eta_i \langle \bm{A}_i, \bm{H} \rangle\nonumber\\
 \leq& C\sqrt{p} \|\bm{\eta}\|_2\|\bm{H}^{-}\|_F.
 \end{align}
 This holds with probability at least $1-\exp(-cp)$.
% with probability at least $1-\exp(-cp)$.
 %where the inequality is from \cite[Lemma III.1]{Huang}.
Substituting \eqref{e:add1} and \eqref{e:add2} into \eqref{e:AT12},  we obtain that
 $$
 p\theta_-\|\bm{H}^{-}\|_F^2\leq C\sqrt{p}\|\bm{\eta}\|_2\|\bm{H}^{-}\|_F+\sum_{i\in (T_1\cup T_2)^c}\eta_i^2
 $$
holds  with probability at least  $1-2\exp(-cp)$.
This inequality implies
 $$
 \|\bm{H}^{-}\|_F\lesssim \frac{\|\bm{\eta}\|_2}{\sqrt{p}}.
 $$
For the case where $\sharp(T_3 \cup T_4) \geq \frac{p}{2}$, similar to the previous case, using \eqref{imdent1+} in Lemma \ref{lemma3.1}, we can derive that
$$\|\bm{H}^{+}\|_F \lesssim \frac{\|\bm{\eta}\|_2}{\sqrt{p}}$$
with probability at least $1 - 2\exp(-cp)$, where $\bm{H}^{+} = \hat{\bm{X}}^{\ell_2} + \bm{X}_0$.
Combining the results from both cases, we immediately arrive at the desired conclusion.

Now, we proceed to prove \eqref{e:setH}.  For the matrices $\bm{H}^{-} \in \mathbb{R}^{n\times m}$ and $\bm{X}_0\in \mathbb{R}^{n\times m}$, applying Lemma \ref{lemma 3.6},
 there exist matrices $\bm{H}_0\in \mathbb{R}^{n\times m}$ and  $\bm{H}_c\in \mathbb{R}^{n\times m}$  such that
 \begin{align*}%\label{-x0}
 \bm{H}^{-}=\bm{H}_0+\bm{H}_c, \ \ \mathrm{rank}(\bm{H}_0)\leq 2r, \ \ \bm{X}_0\bm{H}_c^{\top}=\bm{0},\ \
\bm{X}_0^{\top}\bm{H}_c=\bm{0},\ \ \langle \bm{H}_0, \bm{H}_c\rangle=0,
\end{align*}
where $r=\mathrm{rank}(\bm{X}_0)$.
Furthermore, we have
\begin{align*}%\label{-x02}
 \|\bm{X}_0\|_{*}\overset{(a)}\geq \|\hat{\bm{X}}^{\ell_2}\|_{*}= \|\bm{X}_0+\bm{H}^{-}\|_{*}\overset{(b)}\geq\|\bm{X}_0+\bm{H}_c\|_{*}- \|\bm{H}_0\|_{*}
 \overset{(c)}=\|\bm{X}_0\|_{*}+\|\bm{H}_c\|_{*}-\|\bm{H}_0\|_{*},
 \end{align*}
where the inequality $(a)$ follows from  \eqref{2nu}, the  inequality $(b)$ is due to the triangle inequality on  $\|\cdot\|_{*}$, and the equality $(c)$ is based on Lemma \ref{lemma3.7} with $\bm{X}_0\bm{H}_c^{\top}=\bm{0}$ and
$\bm{X}_0^{\top}\bm{H}_c=\bm{0}$.
From this, we conclude
$$\|\bm{H}_c\|_{*}\leq\|\bm{H}_0\|_{*}.$$
We  then derive the following bound:
$$
\|\bm{H}^{-}\|_{*}\leq  \|\bm{H}_c\|_{*}+\|\bm{H}_0\|_{*}\leq 2\|\bm{H}_0\|_{*}%\leq 2\sqrt{2\mathrm{rank}(\bm{X}_0)}\|\bm{H}_0\|_{F}
\overset{(a)}\leq 2\sqrt{2r}\|\bm{H}^-\|_{F},
$$
where the inequality $(a)$ uses $\mathrm{rank}(\bm{H}_0)\leq 2r$ and $\langle \bm{H}_0, \bm{H}_c\rangle=0$, with the latter giving $\|\bm{H}^-\|_{F}^2=\|\bm{H}_0\|_F^2+\|\bm{H}_c\|_F^2$ by Lemma \ref{lemma3.7}.
 Consequently, \eqref{e:setH} is established.
\end{proof}
\subsection*{B. Proof of Theorem \ref{thm:ell2}}
\begin{proof}
For the solution $\hat{\bm{X}}^{\ell_*}$ to  \eqref{leastproblemcon},
we have
\begin{align}\label{eq:bdell2}
\||\mathcal{A}(\hat{\bm{X}}^{\ell_*})|-\bm{y}\|_2\leq \varepsilon \  \ \ \  \mathrm{and} \ \ \ \  \|\hat{\bm{X}}^{\ell_*}\|_{*}\leq \|\bm{X}_0\|_{*}.
\end{align}

When applying Lemma \ref{lemma3.1}, we make the specific assignment $\bm{X} := \hat{\bm{X}}^{\ell_*}$. This substitution then defines $\bm{H}^{-}=\hat{\bm{X}}^{\ell_*}-\bm{X}_0$ and $\bm{H}^{+}=\hat{\bm{X}}^{\ell}+\bm{X}_0$, along with the sets $T_1, T_2, T_3$, and $T_4$, all consistent with the lemma's definitions. In this proof, it is implicitly understood that the symbol $\bm{X}$ (as used in Lemma \ref{lemma3.1}) consistently refers to $\hat{\bm{X}}^{\ell_{*}}$.
In accordance with  Remark \ref{remark3.2},  we assume, without loss of generality, that $\sharp(T_1\cup T_2)\geq \frac p 2$.
 Using (\ref{eq:bdell2})  and
\eqref{imdent1} in Lemma \ref{lemma3.1}, we derive
$$
\sum_{i\in T_1\cup T_2} (\langle \bm{A}_i, \bm{H}^{-}\rangle^2+\eta_i^2)
-\sum_{i\in T_1} 2\eta_i\langle \bm{A}_i, \bm{H}^{-}\rangle+\sum_{i\in T_2}2\eta_i\langle \bm{A}_i, \bm{H}^{-}\rangle\leq \varepsilon^2.
$$
This means
\begin{align}\label{upper}
\sum_{i\in T_1\cup T_2} \langle \bm{A}_i, \bm{H}^{-}\rangle^2\leq
\sum_{i\in T_1} 2\eta_i\langle \bm{A}_i, \bm{H}^{-}\rangle-\sum_{i\in T_2}2\eta_i\langle \bm{A}_i, \bm{H}^{-}\rangle+\varepsilon^2.
\end{align}

Following a methodology analogous to that employed for establishing inequalities \eqref{e:add1} and \eqref{e:add2}, we can show that the following two probabilistic bounds hold.
First, with probability at least $1-\exp(-cp)$, provided $p\gtrsim (m+n)r$, we have
\begin{align}\label{e:left1}
 \sum_{i\in T_1\cup T_2} \langle \bm{A}_i, \bm{H}^{-}\rangle^2\geq p\theta_-\|\bm{H}^{-}\|_F^2.
 \end{align}
Second, the inequality
 \begin{align}\label{e:right}
 \sum_{i\in T_1}\eta_i\langle \bm{A}_i, \bm{H}^{-}\rangle-\sum_{i\in T_2}\eta_i\langle \bm{A}_i, \bm{H}^{-}\rangle
 \leq C\sqrt{p} \|\bm{\eta}\|_2\|\bm{H}^{-}\|_F
\end{align}
also holds with probability at least $1-\exp(-cp)$.

 Thus, by substituting \eqref{e:left1}, \eqref{e:right} and $\|\bm{\eta}\|_2\leq \varepsilon$ into
 \eqref{upper},   we obtain
 $$
 p\theta_-\|\bm{H}^{-}\|_F^2\leq C\sqrt{p} \|\bm{H}^{-}\|_F\varepsilon+\varepsilon^2
 $$
 with probability at least  $1-2\exp(-cp)$.
This implies
 $$\|\bm{H}^{-}\|_F\lesssim \frac{\varepsilon}{\sqrt{p}}.$$

For the case where $\sharp(T_3 \cup T_4) \geq \frac{p}{2}$, the proof of $\|\bm{H}^{+}\|_F \lesssim \frac{\varepsilon}{\sqrt{p}}$ with $\bm{H}^{+} = \hat{\bm{X}}^{\ell_2} + \bm{X}_0$  follows a similar approach. %However, we need to use $-\bm{X}_0$ instead of $\bm{X}_0$ in \eqref{-x01} and \eqref{-x02}.
Based on \eqref{imdent1+} in Lemma \ref{lemma3.1}, we can derive that
$$ \|\bm{H}^{+}\|_F \lesssim \frac{\varepsilon}{\sqrt{p}} $$
with probability at least $1 - 2\exp(-cp)$, where $\bm{H}^{+} = \hat{\bm{X}}^{\ell_2} + \bm{X}_0$.
Integrating the outcomes from both cases, we  readily deduce the intended conclusion.
\end{proof}

\subsection*{C. \textit{Proof of Theorem \ref{thm:unell2}}}

\begin{proof}
Given that $\hat{\bm{X}}^{u} \in \mathbb{R}^{n\times m}$ is the solution of \eqref{leastproblem} and $\bm{y}=|\mathcal{A}(\bm{X}_0)|+\bm{\eta}$, we have
$$
\||\mathcal{A}(\hat{\bm{X}}^{u})| - \bm{y}\|_2^2 + \lambda \|\hat{\bm{X}}^{u}\|_{*} \leq \||\mathcal{A}(\bm{X}_0)| - \bm{y}\|_2^2 + \lambda \|\bm{X}_0\|_{*} = \|\bm{\eta}\|_2^2 + \lambda \|\bm{X}_0\|_{*},
$$
which implies that
\begin{equation}\label{eq:lamdaeq}
\||\mathcal{A}(\hat{\bm{X}}^{u})|-\bm{y}\|_2^2\leq
\|\bm{\eta}\|_2^2+\lambda(\|\bm{X}_0\|_{*}-\|\hat{\bm{X}}^{u}\|_{*}).
\end{equation}
To apply Lemma \ref{lemma3.1}, we set $\bm{X} := \hat{\bm{X}}^u$. Subsequently, $\bm{H}^{-} = \hat{\bm{X}}^u - \bm{X}_0$, $\bm{H}^{+} = \hat{\bm{X}}^u + \bm{X}_0$, and the sets $T_1, T_2, T_3, T_4$ are defined as specified in the lemma.
According to Remark \ref{remark3.2},  without loss of generality, we assume that $\sharp(T_1\cup T_2)\geq \frac p 2$.
Substituting  \eqref{eq:lamdaeq} into \eqref{imdent1} in Lemma \ref{lemma3.1}, we have that
\begin{align}\label{e:t12}
\|\mathcal{A}_{T_1\cup T_2} (\bm{H}^{-})\|^2_2\leq&\sum_{i\in T_1} 2\eta_i\langle \bm{A}_i, \bm{H}^{-}\rangle-
\sum_{i\in T_2} 2\eta_i\langle \bm{A}_i, \bm{H}^{-}\rangle
+\sum_{i\in (T_1\cup T_2)^c}
\eta_{i}^2+\lambda(\|\bm{X}_0\|_{*}-\|\hat{\bm{X}}^{u}\|_{*})\nonumber\\
=&
2\langle\mathcal{A}(\bm{H}^{-}),\bm{\eta}_{T_1}\rangle
 -2\langle\mathcal{A}(\bm{H}^{-}),\bm{\eta}_{T_2}\rangle+\|\bm{\eta}_{(T_1\cup T_2)^c}\|_2^2+\lambda(\|\bm{X}_0\|_{*}-\|\hat{\bm{X}}^{u}\|_{*})\nonumber\\
 =& 2\langle\bm{H}^{-},\mathcal{A}^{*}(\bm{\eta}_{T_1})-\mathcal{A}^{*}(\bm{\eta}_{T_2})\rangle+\|\bm{\eta}_{(T_1\cup T_2)^c}\|_2^2
 +\lambda(\|\bm{X}_0\|_{*}-\|\bm{H}^{-}+\bm{X}_0\|_{*}),
\end{align}
where
$\mathcal{A}^{*}:\mathbb{R}^{p}\rightarrow \mathbb{R}^{n\times m}$ is the adjoint operator of the liner map $\mathcal{A}$
such that
$$\mathcal{A}^{*}(\bm{b})=\sum_{i=1}^pb_i\bm{A}_i$$
for all vectors $\bm{b}\in \mathbb{R}^{p}$.
For the index set $T\subset [p]$, \( \bm{\eta}_T \in \mathbb{R}^p \) denotes a vector equal to $\bm{\eta}$ on the index set $T$ and zero elsewhere.
Furthermore, applying Lemma \ref{lemma 3.6} with $\bm{B}:=\bm{X}_0$ and $\bm{C}:=\bm{H}^{-}$,  we obtain the decomposition
 $\bm{H}^{-}=\bm{H}_0+\bm{H}_c$ with $\bm{H}_0, \bm{H}_c \in \mathbb{R}^{n\times m}$
 such that
 \begin{align}\label{-x0}
 \mathrm{rank}(\bm{H}_0)\leq 2r, \ \ \bm{X}_0\bm{H}_c^{\top}=\bm{0},\ \
\bm{X}_0^{\top}\bm{H}_c=\bm{0},\ \ \langle \bm{H}_0,\bm{H}_c\rangle=0,
\end{align}
where $r=\mathrm{rank}(\bm{X}_0)$. A simple observation is that
\begin{align}\label{e:Hnuclear}
 \|\bm{H}^{-}\|_{*}\leq \|\bm{H}_0\|_{*}+\|\bm{H}_{c}\|_{*},  \ \ \ \
 \|\bm{H}_{0}\|_{*}\leq \sqrt{2r}\|\bm{H}_{0}\|_F \leq \sqrt{2r}\|\bm{H}^{-}\|_F.
\end{align}
For the first  term of the right side of \eqref{e:t12},  we claim that
\begin{align}\label{innerproduce2}
\langle\bm{H}^{-},\mathcal{A}^{*}(\bm{\eta}_{T_1})-\mathcal{A}^{*}(\bm{\eta}_{T_2})\rangle
\leq \sqrt{2p}(\|\bm{H}_0\|_{*}+\|\bm{H}_{c}\|_{*})\|\bm{\eta}\|_2
\end{align}
holds with probability at least $1-\exp(-cp)$,  and defer its argument to the end of the proof.
For the third term of the right side of \eqref{e:t12},   we deduce
\begin{align}\label{e:inequalitynuclear}
\|\bm{X}_0\|_{*}-\|\bm{H}^{-}+\bm{X}_0\|_{*}
\leq \|\bm{X}_0\|_{*}-(\|\bm{H}_{c}+\bm{X}_0\|_{*}
-\|\bm{H}_{0}\|_{*})
=\|\bm{H}_{0}\|_{*}-\|\bm{H}_{c}\|_{*}.
\end{align}
where the inequality follows from the triangle inequality on $\|\cdot\|_{*}$,
the equality is based on  Lemma \ref{lemma3.7} and the facts $\bm{X}_0\bm{H}_c^{\top}=\bm{0}$ and $\bm{X}_0^{\top}\bm{H}_c=\bm{0}$ in \eqref{-x0}.
Substituting  \eqref{innerproduce2} and \eqref{e:inequalitynuclear} into \eqref{e:t12},  we have
\begin{align*}%\label{e:t122}
\|\mathcal{A}_{T_1\cup T_2} (\bm{H}^{-})\|^2_2+(\lambda-2\sqrt{2p}\|\bm{\eta}\|_2)\|\bm{H}_{c}\|_{*}
 \leq(2\sqrt{2p}\|\bm{\eta}\|_2 +\lambda)\|\bm{H}_0\|_{*}+\|\bm{\eta}_{(T_1\cup T_2)^c}\|_2^2.
\end{align*}
This implies for $\lambda>2\sqrt{2p}\|\bm{\eta}\|_2$ that
\begin{align}\label{e:t122a}
(\lambda-2\sqrt{2p}\|\bm{\eta}\|_2)\|\bm{H}_c\|_{*}
 \leq(2\sqrt{2p}\|\bm{\eta}\|_2 +\lambda)\|\bm{H}_{0}\|_{*}+\|\bm{\eta}_{(T_1\cup T_2)^c}\|_2^2
\end{align}
and
\begin{align}\label{e:t122}
\|\mathcal{A}_{T_1\cup T_2} (\bm{H})\|^2_2
& \leq(2\sqrt{2p}\|\bm{\eta}\|_2 +\lambda)\|\bm{H}_{0}\|_{*}+\|\bm{\eta}_{(T_1\cup T_2)^c}\|_2^2\nonumber\\
 &\overset{(a)}\leq\sqrt{2r}(2\sqrt{2p}\|\bm{\eta}\|_2 +\lambda)\|\bm{H}^{-}\|_{F}+\|\bm{\eta}_{(T_1\cup T_2)^c}\|_2^2,
\end{align}
where $(a)$ follows from \eqref{e:Hnuclear}.
For the inequality \eqref{e:t122a},  using the facts \eqref{e:Hnuclear},
we derive that
\begin{align}\label{e:h}
\|\bm{H}^{-}\|_{*}&\leq
\frac{(2\sqrt{2p}\|\bm{\eta}\|_2 +\lambda)\|\bm{H}_{0}\|_{*}+\|\bm{\eta}_{(T_1\cup T_2)^c}\|_2^2}{\lambda-2\sqrt{2p}\|\bm{\eta}\|_2}+\|\bm{H}_{0}\|_{*}\nonumber\\
 &\leq \frac{2\lambda\sqrt{2r}\|\bm{H}^{-}\|_{F}+\|\bm{\eta}_{(T_1\cup T_2)^c}\|_2^2}{\lambda-2\sqrt{2p}\|\bm{\eta}\|_2}.
\end{align}
Setting the matrix \(\hat{\bm{H}} := \frac{\lambda - 2\sqrt{2p}\|\bm{\eta}\|_2}{2\sqrt{2}\lambda\|\bm{H}^{-}\|_{F} + \frac{\|\bm{\eta}_{(T_1\cup T_2)^c}\|_2^2}{\sqrt{r}}} \bm{H}^{-}\), it follows from \eqref{e:h} that \(\|\hat{\bm{H}}\|_{*} \leq \sqrt{r}\) and \(\|\hat{\bm{H}}\|_F \leq 1\), meaning \(\hat{\bm{H}} \in N_r^{*}:=\{\bm{H}\in \mathbb{R}^{n\times m}:\|\bm{H}\|_F\leq 1, \|\bm{H}\|_{*}\leq \sqrt{r}\}\).
Therefore, from Lemma \ref{rmk1}, it follows that
$$\|\mathcal{A}_{T_1\cup T_2} (\hat{\bm{\bm{H}}})\|^2_2\geq p\theta_{-}\|\hat{\bm{H}}\|_F^2
$$
with probability at least $1-\exp(-cp)$.
This implies that
\begin{align}\label{e:inequalityright}
\|\mathcal{A}_{T_1\cup T_2} (\bm{\bm{H}}^{-})\|^2_2\geq p\theta_{-}\|\bm{H}^{-}\|_F^2.
\end{align}
Substituting  \eqref{e:inequalityright} into \eqref{e:t122},  we have
$$p\theta_{-}\|\bm{H}^{-}\|_F^2\leq\sqrt{2r}(2\sqrt{2p}\|\bm{\eta}\|_2 +\lambda)\|\bm{H}^{-}\|_{F}+\|\bm{\eta}_{(T_1\cup T_2)^c}\|_2^2.
$$
Therefore, we obtain the desired result
$$\|\bm{H}^{-}\|_F \lesssim \frac{\lambda \sqrt{r}}{p} +\frac{\|\bm{\eta}\|_2}{\sqrt{p}}
$$
with probability  at least $1-2\exp(-cp)$.

When \(\sharp(T_3 \cup T_4) \geq \frac{p}{2}\), the proof for the bound \(\|\bm{H}^{+}\|_F \lesssim \frac{\lambda \sqrt{r}}{p} + \frac{|\bm{\eta}|_2}{\sqrt{p}}\) (with \(\bm{H}^{+} = \hat{\bm{X}}^{u} + \bm{X}_0\)) is analogous to the case where \(\sharp(T_1 \cup T_2) \geq \frac{p}{2}\). The only modifications required are to use \(-\bm{X}_0\) instead of \(\bm{X}_0\) in \eqref{-x0} and to use \eqref{imdent1+} in place of \eqref{imdent1}. Combining the results from both cases directly yields the desired conclusion.

It remains  to prove \eqref{innerproduce2}.
Based on the definition of $\mathcal{A}^{*}$ and the fact that $T_1\cap T_2=\varnothing$ as stated in Remark \ref{remark3.2}, there is
\begin{align}\label{e:H-A}
\langle\bm{H}^{-},\mathcal{A}^{*}(\bm{\eta}_{T_1})-\mathcal{A}^{*}(\bm{\eta}_{T_2})\rangle
= &\langle\bm{H}^{-},\mathcal{A}^{*}(\bm{\eta}_{T_1}-\bm{\eta}_{T_2})\rangle
\leq \|\bm{H}^{-}\|_{*}\|\mathcal{A}^{*}(\bm{\eta}_{T_1}-\bm{\eta}_{T_2})\|,
\end{align}
where the inequality is from that the dual norm of the operator norm $\|\cdot\|$ is the nuclear norm $\|\cdot\|_{*}$ \cite[Proposition 2.1]{recht2010guaranteed}.
Recall  the fact for the operator norm $\|\cdot\|$  that
\begin{align}\label{oparetornorm}
\|\mathcal{A}^{*}(\bm{\eta}_{T_1}-\bm{\eta}_{T_2})\|=
\langle \bm{u}\bm{v}^{*}, \mathcal{A}^{*}(\bm{\eta}_{T_1}-\bm{\eta}_{T_2})\rangle
=\sum_{i\in T_1\cup T_2 }\langle \bm{A}_{i},\bm{u}\bm{v}^{*}\rangle \eta_i
\leq  \|\mathcal{A}(\bm{u}\bm{v}^{*})\|_2\|\bm{\eta}_{T_1\cup T_2}\|_2
\overset{(a)}\leq \sqrt{2p}\|\bm{\eta}_{T_1\cup T_2}\|_2
\end{align}
 with high probability at least $1-\exp(-cp)$, where  $\bm{u}$ and $\bm{v}$ are the left and right singular vectors of
the matrix $\mathcal{A}^{*}(\bm{\eta}_{T_1}-\bm{\eta}_{T_2})$ corresponding on the top singular value.
The  inequality $(a)$ above follows from Lemma \ref{spaceSRIP} with $\bm{X}=\bm{u}\bm{v}^{*}$, $\mathrm{rank}{(\bm{u}\bm{v}^{*})}=1$,
$\|\bm{u}\bm{v}^{*}\|_F=1$ and $\theta_{+}\in (0,2)$.

Substituting  \eqref{oparetornorm}, $\|\bm{H}^{-}\|_{*}\leq \|\bm{H}_0\|_{*}+\|\bm{H}_{c}\|_{*}$ and   the fact that $\|\bm{\eta}_{T_1\cup T_2}\|_2\leq \|\bm{\eta}\|_2$ with $T_1\cap T_2=\emptyset$ into \eqref{e:H-A},  we have
\begin{align*}%\label{innerproduce2}
\langle\bm{H}^{-},\mathcal{A}^{*}(\bm{\eta}_{T_1})-\mathcal{A}^{*}(\bm{\eta}_{T_2})\rangle
\leq \sqrt{2p}(\|\bm{H}_0\|_{*}+\|\bm{H}_{c}\|_{*})\|\bm{\eta}\|_2,
\end{align*}
which is \eqref{innerproduce2}.
\end{proof}

\appendix

\section*{Appendix}

\section{Proof of Lemma \ref{lemma3.1}}\label{prof:lemma3.1}

\begin{proof}
Given that $T_1 \cup T_2 \subseteq [p]$,  $\bm{y}=|\mathcal{A}(\bm{X}_0)|+\bm{\eta}$ and $\bm{H}^{-}=\bm{X}-\bm{X}_0$, we can derive the following inequality
\begin{align*}
\||\mathcal{A}({\bm{X}})|-\bm{y}\|_2^2
\geq &\sum_{i\in T_1} (\langle \bm{A}_i, {\bm{X}}\rangle-\langle \bm{A}_i, \bm{X}_0\rangle-\eta_i)^2
+\sum_{i\in T_2} (-\langle \bm{A}_i, {\bm{X}}\rangle+\langle \bm{A}_i, \bm{X}_0\rangle-\eta_i)^2\nonumber\\
=&\sum_{i\in T_1} (\langle \bm{A}_i, \bm{H}^{-}\rangle-\eta_i)^2
+\sum_{i\in T_2} (\langle \bm{A}_i, \bm{H}^{-}\rangle+\eta_i)^2\nonumber\\
=&\sum_{i\in T_1\cup T_2} (\langle \bm{A}_i, \bm{H}^{-}\rangle^2+\eta_i^2)
-\sum_{i\in T_1} 2\eta_i\langle \bm{A}_i, \bm{H}^{-}\rangle+\sum_{i\in T_2}2\eta_i\langle \bm{A}_i, \bm{H}^{-}\rangle,
\end{align*}
which corresponds to equality \eqref{imdent1}.
Similarly, for $\bm{H}^{+}=\bm{X}+\bm{X}_0$, we obtain
\begin{align}%\label{imdent1}
\||\mathcal{A}(\bm{X})|-\bm{y}\|_2^2
\geq &\sum_{i\in T_3} (\langle \bm{A}_i, \bm{X}\rangle+\langle \bm{A}_i, \bm{X}_0\rangle-\eta_i)^2
+\sum_{i\in T_4} (-\langle \bm{A}_i, \bm{X}\rangle-\langle \bm{A}_i, \bm{X}_0\rangle-\eta_i)^2\nonumber\\
=&\sum_{i\in T_3} (\langle \bm{A}_i, \bm{H}^{+}\rangle-\eta_i)^2
+\sum_{i\in T_4} (\langle \bm{A}_i, \bm{H}^{+}\rangle+\eta_i)^2\nonumber\\
=&\sum_{i\in T_3\cup T_4} (\langle \bm{A}_i, \bm{H}^{+}\rangle^2+\eta_i^2)
-\sum_{i\in T_3} 2\eta_i\langle \bm{A}_i, \bm{H}^{+}\rangle+\sum_{i\in T_4}2\eta_i\langle \bm{A}_i, \bm{H}^+\rangle,\nonumber
\end{align}
which corresponds to equality \eqref{imdent1+}.
\end{proof}
\section{Proof of Lemma \ref{lem:upper1}} \label{prof:upper1}

Before proving Lemma \ref{lem:upper1}, we introduce a  lemma that is useful for its demonstration.

\begin{lemma}\label{lempro} \cite{vershynin2018high}
 Consider a random vector
$\bm{x}\in N(0, \bm{I}_n)$
and a Lipschitz function $f: \mathbb{R}^n\rightarrow \mathbb{R}$ with constant $\|f\|_{\mathrm{Lip}}:
|f(\bm{x})-f(\bm{y})|\leq \|f\|_{\mathrm{Lip}}\| \bm{x}-\bm{y}\|_2$.
Then for every $t \geq 0$,
we have
$$\mathbb{P }[|f(\bm{x})-\mathbb{E}[f(\bm{x})]| \geq t] \leq 2 \exp\Big(\frac{-ct^2}{\|f\|_{\mathrm{Lip}}^2}\Big).$$
\end{lemma}
Now, we present the  proof of Lemma \ref{lem:upper1}.
\begin{proof}

For any fixed set $T\subset[p]$, we claim that
\begin{align}\label{e:supfixT}
\sup_{\bm{H}\in N_r^{*}}\sum_{i\in T} \eta_i\langle \bm{A}_i, \bm{H}\rangle
\leq \mathbb{E}\Big[\sup_{\bm{H}\in N_r^{*}}\sum_{i\in T} \eta_i\langle \bm{A}_i, \bm{H}\rangle\Big]+C_0\sqrt{p}\|\bm{\eta}\|_2
\end{align}
holds with probability at least $1- \exp(-p(c\cdot C_0^2 -\frac{\ln 2}{p}))$, where $C_0$ and $c$ are constants with
$C_0>0$ and $c  >2/C_0^2$.
Furthermore, we assert
\begin{align}\label{e:expect}
\mathbb{E}\Big[\sup_{\bm{H}\in N_r^{*}}\sum_{i\in T} \eta_i\langle \bm{A}_i, \bm{H}\rangle\Big]
\leq  C\sqrt{p}\|\bm{\eta}\|_2.
\end{align}
The proofs of these two assertions are presented in sections (I) and (II) that follow.
Consequently, for any fixed set $T\subset[p]$, we can establish that
\begin{align*}%\label{e:supfixT}
\sup_{\bm{H}\in N_r^{*}}\sum_{i\in T} \eta_i\langle \bm{A}_i, \bm{H}\rangle
\lesssim \sqrt{p}\|\bm{\eta}\|_2
\end{align*}
holds with probability at least $1- \exp(-p(C_0^2 c-\frac{\ln 2}{p}))$. Given that  the  total number of subsets $T\subset[p]$ is $2^p$ and $c\cdot C_0^2 >2$, we can conclude that
$$\sup_{\bm{H}\in N_r^{*}, T\subset [p]} \sum_{i\in T} \eta_i\langle \bm{A}_i, \bm{H}\rangle\lesssim \sqrt{p} \|\bm{\eta}\|_2$$
 with probability at least $1-\exp(-p(C_0^2 c-\ln 2-\frac{\ln 2}{p}))$. This is the desired result.

(I) Proof of \eqref{e:supfixT}:
For the Gaussian measurement ensemble \( \mathcal{A} \), we define the matrix $\bm{A}$ as follows
$$
\bm{A}=\begin{pmatrix}
         \mathrm{vec}(\bm{A}_1)^{\top} \\
 \mathrm{vec}(\bm{A}_2)^{\top} \\
         \vdots \\
         \mathrm{vec}(\bm{A}_p)^{\top} \\
       \end{pmatrix}_{p\times mn},
$$
where $ \mathrm{vec}(\bm{A}_i)\in \mathbb{R}^{mn\times 1}$ is  a column vector obtained by stacking the columns of the matrix $\bm{A}_i$
 and the entries of the matrix $\bm{A}_i$ are independently drawn from a standard normal distribution, i.e., $\mathrm{vec}(\bm{A}_i) \sim \mathcal{N}(\bm{0}, \bm{I}_{mn})$.
For any fixed set \( T \subset [p] \), we have
\begin{equation}\label{e:1}
\sum_{i \in T} \eta_i \langle \bm{A}_i, \bm{H} \rangle = \langle \bm{\eta}_T, \mathcal{A}(\bm{H}) \rangle = \langle \bm{\eta}_T, \bm{A} \mathrm{vec}(\bm{H}) \rangle = \langle \mathrm{vec}(\bm{H}), \bm{A}^{\top} \bm{\eta}_T \rangle,
\end{equation}
where \( \bm{\eta}_T \in \mathbb{R}^p \) denotes a vector equal to $\bm{\eta}$ on the index set $T$ and zero elsewhere,
and \( \mathrm{vec}(\bm{H}) \in \mathbb{R}^{mn \times 1}\).
Define the function
$$f(\bm{B})=\sup_{\bm{h}\in N_r^v}\langle \bm{h}, \bm{B}^{\top}\bm{\eta}_T\rangle,$$
where $N_r^v$ denotes the set comprising the vectorization of all matrices in  $N_r^{*}$, and  $\bm{B}$ is a Gaussian matrix with i.i.d Gaussian entries, i.e., $B_{i,j}\sim \mathcal{N}(0,1)$.
For any Gaussian matrices $\bm{B}_1$ and $\bm{B}_2$, we have
\begin{equation}\label{eq:flip}
\Big|\sup_{\bm{h}\in N_r^v}\langle \bm{h}, \bm{B}_{1}^\top\bm{\eta}_T\rangle
-\sup_{\bm{h}\in N_r^v}\langle \bm{h}, \bm{B}_{2}^\top\bm{\eta}_T\rangle\Big|
\leq\sup_{\bm{h}\in N_r^v} |\langle (\bm{B}_{1}-\bm{B}_{2}) \bm{h}, \bm{\eta}_T\rangle|
\leq \|\bm{\eta}\|_2\|\bm{B}_{1}-\bm{B}_{2}\|_F,
\end{equation}
 where the second inequality follows from $\|\bm{h}\|_2\leq 1$ and
the fact $\|\bm{B}_{1}-\bm{B}_{2}\|\leq \|\bm{B}_{1}-\bm{B}_{2}\|_F$. Here,
$\|\cdot\|$ is the matrix operator norm (i.e., the matrix spectral norm).
The inequality (\ref{eq:flip}) implies $\|f\|_{\mathrm{Lip}}=\|\bm{\eta}\|_2$.
By Lemma \ref{lempro} with $\|f\|_{\mathrm{Lip}}=\|\bm{\eta}\|_2$,
we obtain
$$\mathbb{P}\Big[\Big|\sup_{\bm{h}\in N_r^v}\langle \bm{h}, \bm{B}^{\top}\bm{\eta}_T\rangle
-\mathbb{E}\Big[\sup_{\bm{h}\in N_r^v}\langle \bm{h}, \bm{B}^{\top}\bm{\eta}_T\rangle\Big]\Big|\geq t\Big]\leq 2\exp(-\frac{ct^2}{\|\bm{\eta}\|_2^2})$$
for every $t\geq0$. This
implies
$$\mathbb{P}\Big[\sup_{\bm{h}\in N_r^v}\langle \bm{h}, \bm{B}^{\top}\bm{\eta}_T\rangle
\geq \mathbb{E}\Big[\sup_{\bm{h}\in N_r^v}\langle \bm{h}, \bm{B}^{\top}\bm{\eta}_T\rangle\Big]+t\Big]\leq 2\exp(-\frac{ct^2}{\|\bm{\eta}\|_2^2})$$
for every $t\geq0$. Setting $t=C_0\sqrt{p}\|\bm{\eta}\|_2$ with $C_0>0$ and $C_0^2 \cdot c>2$, we obtain
$$\sup_{\bm{h}\in N_r^v}\langle \bm{h}, \bm{B}^{\top}\bm{\eta}_T\rangle
\leq \mathbb{E}\Big[\sup_{\bm{h}\in N_r^v}\langle \bm{h}, \bm{B}^{\top}\bm{\eta}_T\rangle\Big]+C_0\sqrt{p}\|\bm{\eta}\|_2,$$
with probability at least $1- \exp(-p(C_0^2 c-\frac{\ln 2}{p}))$.
This leads to the result
\begin{align}\label{e:E}
\sup_{\bm{H}\in N_r^{*}}\langle \mathrm{vec}(\bm{H}), \bm{B}^{\top} \bm{\eta}_T \rangle
\leq \mathbb{E}\Big[\sup_{\bm{H}\in N_r^{*}}\langle \mathrm{vec}(\bm{H}), \bm{B}^{\top}\bm{\eta}_T\rangle\Big]+C_0\sqrt{p}\|\bm{\eta}\|_2
\end{align}
with probability at least $1- \exp(-p(C_0^2 c-\frac{\ln 2}{p}))$.
Let $\bm{B}=\bm{A}$.
 Using the inequality \eqref{e:E} and the identity  \eqref{e:1}, we obtain \eqref{e:supfixT}.

(II) Proof of \eqref{e:expect}:
Using the identity \eqref{e:1},  we derive
\begin{align}\label{e:w}
\mathbb{E}\Big[\sup_{\bm{H}\in N_r^*}\sum_{i\in T} \eta_i\langle \bm{A}_i, \bm{H}\rangle\Big]
=\mathbb{E}\Big[\sup_{\mathrm{vec}(\bm{H})\in N_r^v}\langle\mathrm{vec}(\bm{H}), \bm{A}^{\top} \bm{\eta}_T\rangle\Big]
=\|\bm{\eta}_T\|_2w(N_r^v),
\end{align}
where  %$N_r^v$ denotes the set comprising the vectorization of all matrices in  $N_r^{*}$, and
$w(N_r^v)$ represents the Gaussian width of the set $N_r^v$. For the set $N_r^v$, we establish the following equalities for its Gaussian width
\begin{align*}
w(N_r^v)=\mathbb{E}\Big[\sup_{\mathrm{vec}(\bm{H})\in N_r^v}\langle \mathrm{vec}(\bm{H}), \bm{x}\rangle\Big]
=\mathbb{E}\Big[\sup_{\bm{H}\in N_r^{*}}\langle\bm{H}, \bm{X}\rangle\Big],
\end{align*}
where \( \bm{x} \in \mathbb{R}^{mn} \) is a random vector whose entries are independently drawn from the standard Gaussian distribution, and \( \bm{X} \in \mathbb{R}^{n \times m} \) is the matrix obtained by reshaping the vector $\bm{x}$ into an $n \times m$ matrix.
%Substituting this expression for the Gaussian width into equation \eqref{e:w}, we  derive
%\begin{align}\label{e:w1}
%\mathbb{E}\Big[\sup_{\bm{H}\in N_r^{*}}\sum_{i\in T} \eta_i\langle \bm{A}_i, \bm{H}\rangle\Big]= \|\bm{\eta}_T\|_2
%\mathbb{E}\Big[\sup_{\bm{H}\in N_r^{*}}\langle\bm{H}, \bm{X}\rangle\Big].
%\end{align}
Given that the dual norm of the matrix operator norm $\|\cdot\|$ is the matrix nuclear norm $\|\cdot\|_{*}$
  \cite[Proposition 2.1]{recht2010guaranteed}, we establish the following inequality for any $\bm{H}\in N_r^{*}$:
  $$\langle\bm{H}, \bm{X}\rangle\leq \|\bm{X}\|\|\bm{H}\|_*\leq \sqrt{r}\|\bm{X}\|.$$
Furthermore, we derive
 $$\mathbb{E}\Big[\sup_{\bm{H}\in N_r^{*}}\langle\bm{H}, \bm{X}\rangle\Big]\leq \sqrt{r}\mathbb{E}[\|\bm{X}\|].$$
 Therefore, we obtain
 $$
 w(N_r^v)\leq \sqrt{r}\mathbb{E}[\|\bm{X}\|].
 $$
Substituting this inequality into  \eqref{e:w}, we obtain
\begin{align*}%\label{e:expect}
\mathbb{E}\Big[\sup_{\bm{H}\in N_r^{*}}\sum_{i\in T} \eta_i\langle \bm{A}_i, \bm{H}\rangle\Big]
\leq \sqrt{r}\|\bm{\eta}_T\|_2\mathbb{E}[\|\bm{X}\|]\leq  C (\sqrt{m}+\sqrt{n})\sqrt{r}\|\bm{\eta}\|_2
\leq C\sqrt{p}\|\bm{\eta}\|_2,
\end{align*}
 where the second inequality follows from the fact that $\mathbb{E}[\|\bm{X}\|] \leq C (\sqrt{m} + \sqrt{n})$, as stated in \cite[Exercise 4.4.6]{vershynin2018high}, and the third inequality is derived from the condition $\sqrt{p} \geq (\sqrt{m} + \sqrt{n}) \sqrt{r}$.
 This completes the proof.
\end{proof}
\section{Proof of Lemma \ref{lem:unbound}} \label{prof:unbound}
Before proving Lemma \ref{lem:unbound}, we first introduce Lemmas \ref{lem:subG}--\ref{lem:sub} concerning Gaussian and sub-Gaussian random variables.
 A random variable \(\xi\) is characterized as sub-Gaussian if there exists a constant \(K > 0\) such that for all \(t > 0\), the following inequality holds:
\[\mathbb{P}[|\xi| \geq t] \leq 2 \exp(-Kt^2).\]
The sub-Gaussian norm of \(\xi\), denoted by \(\|\xi\|_{\psi_2}\), is defined as
\[\|\xi\|_{\psi_2} = \inf\{t > 0 : \mathbb{E}[\exp(\xi^2/t^2)] \leq 2\}.\]
For a comprehensive treatment of sub-Gaussian random variables, we refer the reader to Vershynin \cite[Definition 2.5.6]{vershynin2018high}.
 It is noteworthy that \(\|\cdot\|_{\psi_2}\) constitutes a norm on the space of sub-Gaussian random variables, as elucidated in \cite[Exercise 2.5.7]{vershynin2018high}.

\begin{lemma}\label{lem:subG}
Let $\xi$ be a random variable following a standard normal distribution, i.e., $\xi \sim \mathcal{N}(0,1)$. Then the following properties hold:

$\mathrm{(I)}$ The random variable $\xi$ is sub-Gaussian \cite[Exercise 2.5.8 (a)]{vershynin2018high}, with its sub-Gaussian norm bounded by an absolute constant $C_{\psi}$, i.e.,
\[\|\xi\|_{\psi_2} \leq C_{\psi}.\]

$\mathrm{(II)}$ For each $q \geq 1$, we have \cite[Exercise 2.5.1]{vershynin2018high}:
\[(\mathbb{E}[|\xi|^q])^{\frac{1}{q}} = \sqrt{2}\left(\frac{\Gamma((q+1)/2)}{\Gamma(1/2)}\right)^{\frac{1}{q}}.\]

$\mathrm{(III)}$  The moment generating function of the random variable \(\xi\), as given by \cite[(2.12)]{vershynin2018high}, is
$$\mathbb{E}[\exp(\lambda \xi)] = \exp(\lambda^2/2)$$
for all \(\lambda \in \mathbb{R}\).

\end{lemma}
\begin{lemma} \label{lem:eqinequa}\cite[Proposition 2.5.2]{vershynin2018high}:
For a random variable $\xi$ with $\mathbb{E}[\xi]=0$,
the moment generating function $\mathbb{E}[\exp(\lambda \xi)] = \exp(K_1^2\lambda^2)$ for all $\lambda\in \mathbb{R}$
 are equivalent to the tail  probability bound \(\mathbb{P}[|\xi| \geq t] \leq 2 \exp(-K_2t^2)\) for all $t\geq 0$,
where the parameters $K_1, K_2>0$.
\end{lemma}

\begin{lemma}\label{lem:HTI}\cite[Theorem 2.6.3]{vershynin2018high}
Let $x_1, x_2, \ldots, x_p$ be independent, mean zero, sub-Gaussian  variables,
and $\bm{a}=(a_1,\ldots, a_p)^{\top}\in \mathbb{R}^p$.
Then, for every $t\geq 0$, we have
$$
\mathbb{P}\Big[\Big|\sum_{i=1}^{p} a_ix_i\Big|\geq t\Big]\leq 2\exp\left(-\frac{C_1t^2}{K^2\|\bm{a}\|_2^2}\right),
$$
where $K=\max_{l\in [p]}\|x_l\|_{\psi_2}$ and $C_1$ is an absolute constant.
\end{lemma}

\begin{lemma}\label{lem:sub}\cite[Lemma 2.6.3]{vershynin2018high}
If $\xi$ is a sub-Gaussian variable,  then $\xi-\mathbb{E}\xi$
is sub-gaussian and
$$\|\xi-\mathbb{E}\xi\|_{\psi_2}\leq \sqrt{C_2}\|\xi\|_{\psi_2}$$
where $C_2$ is an absolute constant.
\end{lemma}
\begin{proof}[Proof of Lemma \ref{lem:unbound}]
Let $\overline{U_r} \subset {U}_r:=\{\bm{H}\in \mathbb{R}^{n\times m}:\|\bm{H}\|_F=1, \mathrm{rank}(\bm{H})\leq r\}$ be an
$\epsilon$-net of $U_r$ with respect to the Frobenius norm.
For any $\bm{X}\in U_r$, there exists  $\bm{X}^{\epsilon}\in \overline{U_r}$ such that  $\|\bm{X}-\bm{X}^{\epsilon}\|_F\leq \epsilon$.
Applying the triangle inequality to the absolute value, we obtain
\begin{align*}
\bigg|\Big|\sum_{i=1}^p\eta_i|\langle  \bm{A}_{i},\bm{X}\rangle|\Big|
-\Big|\sum_{i=1}^p\eta_i|\langle  \bm{A}_{i},\bm{X}^{\epsilon}\rangle|\Big|\bigg|
\leq&\Big|\sum_{i=1}^p\eta_i(|\langle  \bm{A}_{i},\bm{X}\rangle|
-|\langle  \bm{A}_{i},\bm{X}^{\epsilon}\rangle|)\Big|\\
\leq& \sum_{i=1}^p|\eta_i|\Big||\langle  \bm{A}_{i},\bm{X}\rangle|
-|\langle  \bm{A}_{i},\bm{X}^{\epsilon}\rangle|\Big|\\
\leq& \sum_{i=1}^p|\eta_i||\langle  \bm{A}_{i},\bm{X}-\bm{X}^{\epsilon}\rangle|.
%\lesssim& \sqrt{(m+n+1)r\log p}\|\bm{\eta}\|_2+|\sum_{i=1}^p\eta_i|,
\end{align*}
That is
\begin{align}\label{e:add3}
\Big|\sum_{i=1}^p\eta_i|\langle  \bm{A}_{i},\bm{X}^{\epsilon}\rangle|\Big|-\sum_{i=1}^p|\eta_i||\langle  \bm{A}_{i},\bm{X}-\bm{X}^{\epsilon}\rangle|
\leq \Big|\sum_{i=1}^p\eta_i|\langle  \bm{A}_{i},\bm{X}\rangle|\Big|
\leq
\Big|\sum_{i=1}^p\eta_i|\langle  \bm{A}_{i},\bm{X}^{\epsilon}\rangle|\Big|+\sum_{i=1}^p|\eta_i||\langle  \bm{A}_{i},\bm{X}-\bm{X}^{\epsilon}\rangle|.
\end{align}
For the matrix $\bm{X}-\bm{X}^{\epsilon}$ with rank at most $2r$, provided \(p \geq (m + n) r\),
 the following inequality holds with probability at least \(1 - \exp(-cp)\):
\begin{align}\label{e:etaHab}
\sum_{i=1}^p |\eta_i| |\langle \bm{A}_{i}, \bm{X}-\bm{X}^{\epsilon} \rangle|
\leq \| \bm{\eta} \|_2 \| \mathcal{A}(\bm{X}-\bm{X}^{\epsilon}) \|_2 \leq \theta_{+} \sqrt{p} \| \bm{\eta} \|_2\|\bm{X}-\bm{X}^{\epsilon}\|_2\leq
\theta_{+}\sqrt{p} \| \bm{\eta} \|_2\epsilon,
\end{align}
where the second inequality follows from Lemma \ref{spaceSRIP}.

%\begin{align}\label{e:etaHab}
%_{i=1}^p|\eta_i||\langle  \bm{A}_{i},\bm{X}_{2r}\rangle|\leq C\sqrt{p}\|\bm{\eta}\|_2
%\end{align}
%holds with probability at least $1-2e^{-cp}$, provided $p\geq(\sqrt{m}+\sqrt{n})^2r$.
For any $\bm{X}^{\epsilon}\in \overline{U_r}$, we claim that
\begin{align}\label{e:Aeatuppera1}
\Big|\sum_{i=1}^p\eta_i|\langle  \bm{A}_{i},\bm{X}^{\epsilon}\rangle|\Big|
%\leq t+
%\Big|\sum_{i=1}^p\eta_i\mathbb{E}[|\langle \bm{A}_{i},\bm{X}\rangle|]\Big|
\leq t+\sqrt{\frac{2}{\pi}}\Big|\sum_{i=1}^p\eta_i\Big|
\end{align}
and
\begin{align}\label{e:Aeatuppera2}
\Big|\sum_{i=1}^p\eta_i|\langle  \bm{A}_{i},\bm{X}^{\epsilon}\rangle|\Big|
\geq \sqrt{\frac{2}{\pi}}\Big|\sum_{i=1}^p\eta_i\Big|-t
\end{align}
each holding with probability  at least $1- 2\exp(-\frac{C_1t^2}{C_2C_{\psi}^2\|\bm{\eta}\|_2^2}+ (m+n+1)r\log ({9}/{\epsilon}))$.
We will provide the arguments of \eqref{e:Aeatuppera1} and \eqref{e:Aeatuppera2} at the end of the proof.

Substituting  \eqref{e:Aeatuppera1} and \eqref{e:etaHab} into the right-hand of  inequality \eqref{e:add3}, we obtain that
\begin{align}\label{e:upperbound}
\Big|\sum_{i=1}^p\eta_i|\langle  \bm{A}_{i},\bm{X}\rangle|\Big|
\leq t+\sqrt{\frac{2}{\pi}}\Big|\sum_{i=1}^p\eta_i\Big|
+\theta_{+}\sqrt{p}\|\bm{\eta}\|_2\epsilon
\end{align}
holds with probability at least $1-\exp(-cp)-2\exp(-\frac{C_1t^2}{C_2C_{\psi}^2\|\bm{\eta}\|_2^2}+ (m+n+1)r\log ({9}/{\epsilon}))$.
Similarly, for the left-hand of  the  inequality \eqref{e:add3}, using  \eqref{e:Aeatuppera2} and \eqref{e:etaHab},  we have
% where(a) and (b) are due to  \eqref{e:Aeatuppera1} and \eqref{e:etaHab}, respectively. Similarly, we also obtain
\begin{align}\label{e:lowerbound}
\Big|\sum_{i=1}^p\eta_i|\langle  \bm{A}_{i},\bm{X}\rangle|\Big|
%&\geq\Big|\sum_{i=1}^p\eta_i|\langle  \bm{A}_{i},\bm{X}^{\epsilon}\rangle|\Big|+\sum_{i=1}^p|\eta_i|\cdot|\langle  \bm{A}_{i},\bm{X}-\bm{X}^{\epsilon}\rangle|\\
%&\geq \sqrt{\frac{2}{\pi}}\Big|\sum_{i=1}^p\eta_i\Big|-t-\sum_{i=1}^p|\eta_i|\cdot|\langle  \bm{A}_{i},\bm{X}-\bm{X}^{\epsilon}\rangle|\nonumber\\
&\geq\sqrt{\frac{2}{\pi}}\Big|\sum_{i=1}^p\eta_i\Big|-t
-\theta_{+}\sqrt{p}\|\bm{\eta}\|_2\epsilon
%\lesssim& \sqrt{(m+n+1)r\log p}\|\bm{\eta}\|_2+|\sum_{i=1}^p\eta_i|,
\end{align}
with probability at least $1-\exp(-cp)-2\exp(-\frac{C_1t^2}{C_2C_{\psi}^2\|\bm{\eta}\|_2^2}+ (m+n+1)r\log ({9}/{\epsilon}))$.

Now, we are ready to prove the results $\mathrm{(I)}$ and $\mathrm{(II)}$.

$\mathrm{(I)}$ Taking $t= \sqrt{\frac{3(m+n+1)rC_2\log p}{2C_1}}C_{\psi}\|\bm{\eta}\|_2$ and $\epsilon=9/\sqrt{p}$ in \eqref{e:upperbound}, we obtain that
\begin{align*}
\Big|\sum_{i=1}^p\eta_i|\langle  \bm{A}_{i},\bm{X}\rangle|\Big|
\leq&  \sqrt{\frac{3(m+n+1)rC_2\log p}{2C_1}}C_{\psi}\|\bm{\eta}\|_2+\sqrt{\frac{2}{\pi}}\Big|\sum_{i=1}^p\eta_i\Big|
+9\theta_{+}\|\bm{\eta}\|_2\\
\lesssim& \sqrt{(m+n+1)r\log p}\|\bm{\eta}\|_2+\Big|\sum_{i=1}^p\eta_i\Big|
\end{align*}
holds with probability at least $1- \exp(-cp)-2\exp(-(m+n+1)r\log p)$.

$\mathrm{(II)}$ Substituting  the condition $\Big|\sum_{i=1}^p\eta_i\Big|\geq C_0\sqrt{p}\|\bm{\eta}\|_2$ with $C_0\in (0,1)$ into \eqref{e:lowerbound},  we obtain
that
\begin{align*}%\label{e:lowerbound}
\Big|\sum_{i=1}^p\eta_i|\langle  \bm{A}_{i},\bm{X}\rangle|\Big|
\geq C_0\sqrt{\frac{ 2 p}{\pi}}\|\bm{\eta}\|_2-t-\theta_{+}\sqrt{p}\|\bm{\eta}\|_2\epsilon
%\lesssim& \sqrt{(m+n+1)r\log p}\|\bm{\eta}\|_2+|\sum_{i=1}^p\eta_i|,
\end{align*}
 holds with a probability of at least $1-\exp(-cp)-2\exp(-\frac{C_1t^2}{C_2C_{\psi}^2\|\bm{\eta}\|_2^2}+ (m+n+1)r\log ({9}/{\epsilon}))$.
We can take $t= C_0\sqrt{\frac{p}{2\pi}}\|\bm{\eta}\|_2$ and $\epsilon<\min\Big\{\frac{C_0}{\sqrt{2\pi}C}, 9\exp(-\frac{C_1C_0^2}{2\pi C_2C_{\psi}^2})\Big\}$ in \eqref{e:lowerbound}. This leads to
$$\Big|\sum_{i=1}^p\eta_i|\langle  \bm{A}_{i},\bm{X}\rangle|\Big|
\gtrsim\sqrt{p}\|\bm{\eta}\|_2,
$$
which holds with a probability of at least $1-3\exp(-cp)$.

%For  any $\bm{X}_{2r}\in U_{2r}$,
%define the index sets
%\[
%S_1:=\{i\in [p]:\mathrm{sign}(\eta_i \langle  \bm{A}_{i},\bm{X}_{2r}\rangle)= 1\} \text{ and } S_2:=\{i\in [p]:\mathrm{sign}(\eta_i \langle  \bm{A}_{i},\bm{X}_{2r}\rangle)=-1\}.
%\]
%Then
%$$
%\sum_{i=1}^p|\eta_i||\langle  \bm{A}_{i},\bm{X}_{2r}\rangle|
%=\sum_{i\in S_1}\eta_i\langle  \bm{A}_{i},\bm{X}_{2r}\rangle+\sum_{i\in S_2}\eta_i\langle  \bm{A}_{i},-\bm{X}_{2r}\rangle.
%$$
%Consequently,
%$$
%\sum_{i=1}^p|\eta_i||\langle  \bm{A}_{i},\bm{X}_{2r}\rangle|\leq \sup_{T\subset [p]}\eta_i\langle  \bm{A}_{i},\bm{X}_{2r}\rangle
%+\sup_{T\subset [p]}\eta_i\langle  \bm{A}_{i},-\bm{X}_{2r}\rangle\leq 2 \sup_{\bm{X}_{2r}\in U_{2r}, T\subset [p]}\sum_{i\in T} \eta_i\langle \bm{A}_i, \bm{X}_{2r}\rangle
%$$
%where the last step uses that $-\bm{X}_{2r} \in U_{2r}$.
%By Lemma \ref{lem:upper1}, the   supremum on the right is bounded as
%$$
%\sup_{\bm{X}_{2r}\in U_{2r}, T\subset [p]}\sum_{i\in T} \eta_i\langle \bm{A}_i, \bm{X}_{2r}\rangle
%\leq C\sqrt{p}\|\bm{\eta}\|_2
%$$
%with probability at least $1-2e^{-cp}$, provided $p\geq(\sqrt{m}+\sqrt{n})^2r$.
%Combining the inequalities yields
%$$\sum_{i=1}^p |\eta_i||\langle  \bm{A}_{i},\bm{X}_{2r}\rangle| \leq C\sqrt{p}\|\bm{\eta}\|_2$$
 %for all $\bm{X}_{2r}\in U_{2r}$, which is exactly \eqref{e:etaHab}.

Now, we prove \eqref{e:Aeatuppera1} and \eqref{e:Aeatuppera2}.
%we establish the lower and upper bounds of $\Big|\sum_{i=1}^p\eta_i|\langle  \bm{A}_{i},\bm{X}\rangle|\Big|$ for any $\bm{X}\in U_r$ applying Lemma \ref{cover}.
 %and \ref{lem:subGnorm}.
We denote the $(k,l)$-th entries of $\bm{X}$ and $\bm{A}_i$ by $X_{kl}$ and $(A_i)_{kl}$, respectively.
For any fixed $\bm{X}\in U_r$, %using Lemma \ref{lem:subGnorm} and
since each elements $(A_i)_{kl}$ of $\bm{A}_i$ is  independently drawn from $\mathcal{N}(0,1)$,
then $X_{kl}\cdot (A_i)_{kl}\sim \mathcal{N}(0,X_{kl}^2)$, which implies
$\langle \bm{A}_{i},\bm{X}\rangle\sim \mathcal{N}(0,1)$. Here, we use $\|\bm{X}\|_F=1$.
%implying %$\|(\bm{A}_i)_{kl}\|_{\psi_2}\leq C_{\psi}$,
Moreover, for all $i \in [p]$, we have
\begin{align}\label{e:etaA}
\|\langle \bm{A}_{i},\bm{X}\rangle\|_{\psi_2} \overset{(a)}{\leq} C_{\psi}, \quad
\mathbb{E}[|\langle \bm{A}_{i},\bm{X}\rangle|]\overset{(b)}{=}\sqrt{2}\frac{\Gamma(1)}{\Gamma(1/2)}=\sqrt{\frac{2}{\pi}},
\quad \mathbb{E}[\exp(\lambda \langle \bm{A}_{i},\bm{X}\rangle)] \overset{(c)}= \exp(\lambda^2/2)\ \  \text{for all } \lambda \in \mathbb{R},
\end{align}
where $C_{\psi}$ is an absolute constant. Here, ($a$), ($b$) and ($c$) follow from Lemma \ref{lem:subG} (I), (II) and (III), respectively.
Combining the third equality in \eqref{e:etaA} with Lemma \ref{lem:eqinequa}, we have
$$
\mathbb{P}[|\langle \bm{A}_{i},\bm{X}\rangle|\geq t]\leq 2\exp(-K_2t^2)
$$
for all $t\geq 0$, where $K_2>0$.
Furthermore, by the definition of a sub-Gaussian variable, it is  evident that  that the random variable $|\langle \bm{A}_i, \bm{X} \rangle|$ is sub-Gaussian.
Consequently, by the symmetry of $\langle \bm{A}_{i},\bm{X}\rangle\sim \mathcal{N}(0,1)$,  we have
 $$\||\langle \bm{A}_{i},\bm{X}\rangle|\|_{\psi_2} =\|\langle \bm{A}_{i},\bm{X}\rangle\|_{\psi_2} \leq C_{\psi}.$$
 Combining this with Lemma \ref{lem:sub}, it is evident that \(|\langle \bm{A}_i, \bm{X} \rangle| - \mathbb{E}[|\langle \bm{A}_i, \bm{X} \rangle|]\) is also sub-Gaussian and
$$
\||\langle \bm{A}_{i},\bm{X}\rangle|-\mathbb{E}[|\langle \bm{A}_{i},\bm{X}\rangle|]\|_{\psi_2}\leq \sqrt{C_2}\||\langle \bm{A}_{i},\bm{X}\rangle|\|_{\psi_2}\leq
\sqrt{C_2}C_{\psi}.
$$
%where the last inequality follows from the first inequality of \eqref{e:etaA}.
For the sub-Gaussian variables $|\langle \bm{A}_{i},\bm{X}\rangle|-\mathbb{E}[|\langle \bm{A}_{i},\bm{X}\rangle|]$ with $i\in [p]$,
 by Lemma \ref{lem:HTI} with $t\geq 0$, one has
\begin{align}\label{e:lem4e12}
\mathbb{P}\Big[\Big|\sum_{i=1}^p\eta_i(|\langle \bm{A}_{i},\bm{X}\rangle|-\mathbb{E}[|\langle \bm{A}_{i},\bm{X}\rangle|])\Big|
\geq t\Big]
\leq 2\exp\left(-\frac{C_1t^2}{C_2C_{\psi}^2\|\bm{\eta}\|_2^2}\right).
\end{align}
For the term  $\Big|\sum_{i=1}^p\eta_i\mathbb{E}[|\langle \bm{A}_{i},\bm{X}\rangle|]\Big|$, we have
\begin{align}\label{e:meanupper}
%\frac{1}{3}\Big|\sum_{i=1}^p\eta_i\Big| \leq
\Big|\sum_{i=1}^p\eta_i\mathbb{E}[|\langle \bm{A}_{i},\bm{X}\rangle|]\Big|=\sqrt{\frac{2}{\pi}}\Big|\sum_{i=1}^p\eta_i\Big|,
\end{align}
which follows from  the second  equality in \eqref{e:etaA}.
Substituting \eqref{e:meanupper} into \eqref{e:lem4e12}, we obtain
\begin{align}\label{e:lem4e1}
\mathbb{P}\bigg[\bigg|\Big|\sum_{i=1}^p\eta_i|\langle \bm{A}_{i},\bm{X}\rangle|\Big|-
\sqrt{\frac{2}{\pi}}\Big|\sum_{i=1}^p\eta_i\Big|\bigg|
\geq t\bigg]
\leq 2\exp\left(-\frac{C_1t^2}{C_2C_{\psi}^2\|\bm{\eta}\|_2^2}\right)
\end{align}
for all $t>0$.

%, where  $C_1$, $C_2$ and $C_{\psi}$ are absolute constants.
The inequality \eqref{e:lem4e1} directly yields the following probabilistic bounds:
\begin{align}\label{e:proeta1}
\mathbb{P}\Big[\Big|\sum_{i=1}^p\eta_i|\langle \bm{A}_{i},\bm{X}\rangle|\Big|
\leq t+\sqrt{\frac{2}{\pi}}\Big|\sum_{i=1}^p\eta_i\Big|\Big]
\geq 1-2\exp\left(-\frac{C_1t^2}{C_2C_{\psi}^2\|\bm{\eta}\|_2^2}\right)
\end{align}
and
\begin{align}\label{e:proeta2}
\mathbb{P}\Big[\Big|\sum_{i=1}^p\eta_i|\langle \bm{A}_{i},\bm{X}\rangle|\Big|
\geq \sqrt{\frac{2}{\pi}}\Big|\sum_{i=1}^p\eta_i\Big|-t\Big]
\geq1- 2\exp\left(-\frac{C_1t^2}{C_2C_{\psi}^2\|\bm{\eta}\|_2^2}\right).
\end{align}

Let  $\overline{U_r} \subset {U}_r$ be an $\epsilon$-net
of $U_r$ in the Frobenius norm as given by Lemma \ref{cover},
 with cardinality $\#\overline{U_r}\leq ({9}/{\epsilon})^{(m+n+1)r}$.
Applying \eqref{e:proeta1} and taking a union bound over $\overline{U_r}$, we obtain that,
with probability at least $1- 2\exp(-\frac{C_1t^2}{C_2C_{\psi}^2\|\bm{\eta}\|_2^2}+ (m+n+1)r\log ({9}/{\epsilon}))$,
the following inequality holds:
\begin{align*}%\label{e:Aeatuppera1}
\Big|\sum_{i=1}^p\eta_i|\langle  \bm{A}_{i},\bm{X}\rangle|\Big|
%\leq t+
%\Big|\sum_{i=1}^p\eta_i\mathbb{E}[|\langle \bm{A}_{i},\bm{X}\rangle|]\Big|
\leq t+\sqrt{\frac{2}{\pi}}\Big|\sum_{i=1}^p\eta_i\Big|
\end{align*}
which is exactly \eqref{e:Aeatuppera1}.

 Similarly,
by \eqref{e:proeta1} and the same union bound over
the $\epsilon$-net $\overline{U_r}$, we also have
\begin{align*}%\label{e:Aeatuppera2}
\Big|\sum_{i=1}^p\eta_i|\langle  \bm{A}_{i},\bm{X}\rangle|\Big|
\geq \sqrt{\frac{2}{\pi}}\Big|\sum_{i=1}^p\eta_i\Big|-t
\end{align*}
with probability at least $1- 2\exp\left(-\frac{C_1t^2}{C_2C_{\psi}^2\|\bm{\eta}\|_2^2}+ (m+n+1)r\log \left(\frac{9}{\epsilon}\right)\right)$.
This is  exactly\eqref{e:Aeatuppera2}.
\end{proof}
%\bibliographystyle{plain}
%\bibliography{references}

\end{document}